\documentclass[journal]{IEEEtran}
\usepackage{amsmath,amssymb,amsthm,bm,booktabs,graphicx,cite,multirow,algorithm,algorithmic,microtype}
\usepackage[hidelinks]{hyperref}
\newcommand{\chkIdentity}{$3.6\cdot10^{-15}$}
\newcommand{\chkMinor}{$1.7\cdot10^{-13}$}

\newcommand{\chkWMMSE}{$2.4\cdot10^{-11}$}
\newcommand{\chkXMMSE}{$1.2\cdot10^{-15}$}
\newcommand{\chkBfGap}{$1.2\cdot10^{-14}$}
\newcommand{\chkRateErr}{$3.4\cdot10^{-3}$}

\newcommand{\chkRateN}{10}
\newcommand{\chkRateFP}{0.991}
\newcommand{\chkRateSEFP}{0.898}
\newcommand{\chkRateICT}{0.763}

\newcommand{\pcCellSINR}{1.5}
\newcommand{\pcCellSINRp}{33.7}
\newcommand{\pcCellFPitNN}{406}
\newcommand{\pcCellFPitNNN}{1704}
\newcommand{\pcCellFPrelBest}{-0.77}
\newcommand{\pcCellFPwithin}{61.4}
\newcommand{\pcCellFPatFifty}{-4.53}

\newcommand{\pcCellFPrising}{67.6}
\newcommand{\pcCellFPkkt}{$1.7\cdot10^{-6}$}
\newcommand{\pcCellSEFPitNN}{17}
\newcommand{\pcCellSEFPitNNN}{34}
\newcommand{\pcCellSEFPrelBest}{-0.24}
\newcommand{\pcCellSEFPwithin}{96.8}
\newcommand{\pcCellSEFPatFifty}{-0.68}

\newcommand{\pcCellSEFPkkt}{below $10^{-15}$}
\newcommand{\pcCellICTitNN}{11}
\newcommand{\pcCellICTitNNN}{17}
\newcommand{\pcCellICTrelBest}{-0.21}
\newcommand{\pcCellICTwithin}{97.4}
\newcommand{\pcCellICTatFifty}{-0.44}

\newcommand{\pcCellICTvsFPmean}{+0.62}
\newcommand{\pcCellICTvsFPwins}{38.4}
\newcommand{\pcCellICTvsFPlosses}{2.6}
\newcommand{\pcCellICTvsSEFPmean}{+0.04}
\newcommand{\pcCellICTvsSEFPwins}{0.8}
\newcommand{\pcCellICTvsSEFPlosses}{0.2}

\newcommand{\pcDenseSINR}{15.6}
\newcommand{\pcDenseSINRp}{34.6}
\newcommand{\pcDenseFPitNN}{1962}
\newcommand{\pcDenseFPitNNN}{4350}
\newcommand{\pcDenseFPrelBest}{-0.43}
\newcommand{\pcDenseFPwithin}{37.8}
\newcommand{\pcDenseFPatFifty}{-3.50}

\newcommand{\pcDenseFPrising}{98.6}

\newcommand{\pcDenseSEFPitNN}{16}
\newcommand{\pcDenseSEFPitNNN}{35}
\newcommand{\pcDenseSEFPrelBest}{-0.04}
\newcommand{\pcDenseSEFPwithin}{94.8}
\newcommand{\pcDenseSEFPatFifty}{-0.09}

\newcommand{\pcDenseICTitNN}{10}
\newcommand{\pcDenseICTitNNN}{18}
\newcommand{\pcDenseICTrelBest}{-0.03}
\newcommand{\pcDenseICTwithin}{95.4}
\newcommand{\pcDenseICTatFifty}{-0.04}

\newcommand{\pcDenseICTvsFPmean}{+0.41}
\newcommand{\pcDenseICTvsFPwins}{61.4}
\newcommand{\pcDenseICTvsFPlosses}{4.0}
\newcommand{\pcDenseICTvsSEFPmean}{+0.01}
\newcommand{\pcDenseICTvsSEFPwins}{2.2}
\newcommand{\pcDenseICTvsSEFPlosses}{1.6}

\newcommand{\pcNdrops}{500}
\newcommand{\pcT}{5000}
\newcommand{\pfNdrops}{60}
\newcommand{\pfT}{500}
\newcommand{\pfFP}{98.37}
\newcommand{\pfStdFP}{5.05}
\newcommand{\pfItersFP}{49.7}
\newcommand{\pfCapFP}{98.5}
\newcommand{\pfOffFP}{0.10}
\newcommand{\pfRateFP}{1.82}
\newcommand{\pfSEFP}{102.04}
\newcommand{\pfStdSEFP}{5.34}
\newcommand{\pfItersSEFP}{38.0}
\newcommand{\pfCapSEFP}{27.7}
\newcommand{\pfOffSEFP}{0.00}
\newcommand{\pfRateSEFP}{1.64}
\newcommand{\pfICT}{105.39}
\newcommand{\pfStdICT}{5.56}
\newcommand{\pfItersICT}{18.7}
\newcommand{\pfCapICT}{3.4}
\newcommand{\pfOffICT}{0.00}
\newcommand{\pfRateICT}{1.47}
\newcommand{\pfICTminusFP}{+7.03}
\newcommand{\pfICTminusFPlo}{6.78}
\newcommand{\pfICTminusFPhi}{7.27}
\newcommand{\pfICTminusFPwins}{60}

\newcommand{\pfICTminusSEFP}{+3.35}
\newcommand{\pfICTminusSEFPlo}{3.24}
\newcommand{\pfICTminusSEFPhi}{3.46}
\newcommand{\pfICTminusSEFPwins}{60}

\newcommand{\pfSEFPminusFP}{+3.68}
\newcommand{\pfSEFPminusFPlo}{3.51}
\newcommand{\pfSEFPminusFPhi}{3.85}
\newcommand{\pfSEFPminusFPwins}{60}

\newcommand{\pfGeoFP}{3.23}

\newcommand{\pfGeoSEFP}{3.37}

\newcommand{\pfGeoICT}{3.51}
\newcommand{\ewN}{10{,}000}
\newcommand{\ewFP}{555.6}
\newcommand{\ewBestFP}{1.1}
\newcommand{\ewSEFP}{569.7}
\newcommand{\ewBestSEFP}{6.8}
\newcommand{\ewICT}{570.3}
\newcommand{\ewBestICT}{98.2}
\newcommand{\ewICTvsFP}{2.79}
\newcommand{\ewICTvsFPwins}{98.7}
\newcommand{\ewICTvsFPties}{0.7}
\newcommand{\ewICTvsFPlosses}{0.7}
\newcommand{\ewICTvsSEFP}{0.12}
\newcommand{\ewICTvsSEFPwins}{62.3}
\newcommand{\ewICTvsSEFPties}{36.5}
\newcommand{\ewICTvsSEFPlosses}{1.2}

\newcommand{\nrpN}{10{,}000}
\newcommand{\nrpFP}{698.2}
\newcommand{\nrpBestFP}{0.3}
\newcommand{\nrpSEFP}{762.2}
\newcommand{\nrpBestSEFP}{15.0}
\newcommand{\nrpICT}{781.9}
\newcommand{\nrpBestICT}{91.0}
\newcommand{\nrpICTvsFP}{12.49}
\newcommand{\nrpICTvsFPwins}{98.9}
\newcommand{\nrpICTvsFPties}{0.0}
\newcommand{\nrpICTvsFPlosses}{1.1}
\newcommand{\nrpICTvsSEFP}{2.84}
\newcommand{\nrpICTvsSEFPwins}{71.5}
\newcommand{\nrpICTvsSEFPties}{19.8}
\newcommand{\nrpICTvsSEFPlosses}{8.7}

\newcommand{\convN}{300}
\newcommand{\convFPitNNN}{941}

\newcommand{\convFPlimit}{0.89}

\newcommand{\convSEFPitNNN}{20}

\newcommand{\convICTitNNN}{7}

\newcommand{\budN}{20}

\newcommand{\budFPgain}{+1.47}

\newcommand{\budSEFPgain}{-0.01}

\newcommand{\budICTgain}{+0.00}
\newcommand{\budICTfiftyMinusFPfiveh}{+5.66}
\newcommand{\budICTfiftyMinusFPfivehWins}{20}
\newcommand{\budICTfiftyMinusSEFPfiveh}{+3.43}

\newcommand{\initN}{20}

\newcommand{\initOffWGFP}{16.28}

\newcommand{\initFPdrop}{-42.93}
\newcommand{\initSEFPdrop}{-6.54}
\newcommand{\initICTdrop}{+0.35}
\newcommand{\snrN}{30}

\newcommand{\snrMinWins}{30}
\newcommand{\snrItFPlo}{20.2}
\newcommand{\snrItFPhi}{49.9}
\newcommand{\snrItSEFPlo}{7.7}
\newcommand{\snrItSEFPhi}{44.6}
\newcommand{\snrItICTlo}{4.7}
\newcommand{\snrItICThi}{24.7}

\newcommand{\snrLo}{$-$6}
\newcommand{\snrHi}{14}
\newcommand{\snrGmSEFPlo}{+5.14\%}
\newcommand{\snrGmSEFPhi}{+3.61\%}
\newcommand{\snrGmFPlo}{+12.61\%}
\newcommand{\snrGmFPhi}{+8.28\%}

\newcommand{\snrGmSEFPloB}{+5.57\%}
\newcommand{\snrGmSEFPhiB}{+4.01\%}
\newcommand{\snrGmFPloB}{+13.71\%}
\newcommand{\snrGmFPhiB}{+9.11\%}

\newcommand{\isdN}{200}
\newcommand{\isdFPitNearLo}{1367}
\newcommand{\isdFPitNearHi}{1800}
\newcommand{\isdFPitTwo}{616}
\newcommand{\isdFPitFive}{45}
\newcommand{\isdFPitTen}{5}
\newcommand{\isdICTitMax}{19}
\newcommand{\isdICTitMin}{2}

\newcommand{\isdWithinTen}{99.0}
\newcommand{\bfMfourN}{60}
\newcommand{\bfMfourT}{1500}
\newcommand{\bfMfourSINR}{31.6}
\newcommand{\bfMfourSINRp}{54.1}
\newcommand{\bfMfourEqFPW}{$1.3\cdot10^{-9}$}
\newcommand{\bfMfourEqXR}{$5.2\cdot10^{-11}$}
\newcommand{\bfMfourFPitNN}{777}
\newcommand{\bfMfourFPitNNN}{1390}
\newcommand{\bfMfourFPrelBest}{-11.21}
\newcommand{\bfMfourFPwithin}{5.0}
\newcommand{\bfMfourFPatFifty}{-19.04}
\newcommand{\bfMfourFPatThreeH}{-14.69}
\newcommand{\bfMfourFPrising}{96.7}
\newcommand{\bfMfourFPtime}{0.06}

\newcommand{\bfMfourXMitNN}{14}
\newcommand{\bfMfourXMitNNN}{21}
\newcommand{\bfMfourXMrelBest}{-0.74}
\newcommand{\bfMfourXMwithin}{85.0}
\newcommand{\bfMfourXMatFifty}{-0.97}
\newcommand{\bfMfourXMatThreeH}{-0.74}

\newcommand{\bfMfourXMtime}{0.07}

\newcommand{\bfMfourICTitNN}{12}
\newcommand{\bfMfourICTitNNN}{15}
\newcommand{\bfMfourICTrelBest}{-0.38}
\newcommand{\bfMfourICTwithin}{91.7}
\newcommand{\bfMfourICTatFifty}{-0.62}
\newcommand{\bfMfourICTatThreeH}{-0.38}

\newcommand{\bfMfourICTtime}{0.09}

\newcommand{\bfMfourICTvsFPmean}{+13.38}
\newcommand{\bfMfourICTvsFPwins}{93.3}
\newcommand{\bfMfourICTvsFPlosses}{3.3}

\newcommand{\bfMfourICTvsXMmean}{+0.43}
\newcommand{\bfMfourICTvsXMwins}{13.3}
\newcommand{\bfMfourICTvsXMlosses}{6.7}

\newcommand{\bfMfourICTatFiftyVsFPfinal}{+13.11}

\newcommand{\bfMeightN}{40}
\newcommand{\bfMeightT}{1200}
\newcommand{\bfMeightSINR}{48.8}
\newcommand{\bfMeightSINRp}{63.1}

\newcommand{\bfMeightFPitNN}{840}
\newcommand{\bfMeightFPitNNN}{1160}
\newcommand{\bfMeightFPrelBest}{-26.94}
\newcommand{\bfMeightFPwithin}{0.0}
\newcommand{\bfMeightFPatFifty}{-34.07}
\newcommand{\bfMeightFPatThreeH}{-29.99}
\newcommand{\bfMeightFPrising}{100.0}
\newcommand{\bfMeightFPtime}{0.13}

\newcommand{\bfMeightXMitNN}{14}
\newcommand{\bfMeightXMitNNN}{16}
\newcommand{\bfMeightXMrelBest}{-0.12}
\newcommand{\bfMeightXMwithin}{95.0}
\newcommand{\bfMeightXMatFifty}{-0.30}
\newcommand{\bfMeightXMatThreeH}{-0.12}

\newcommand{\bfMeightXMtime}{0.14}

\newcommand{\bfMeightICTitNN}{8}
\newcommand{\bfMeightICTitNNN}{9}
\newcommand{\bfMeightICTrelBest}{-0.00}
\newcommand{\bfMeightICTwithin}{100.0}
\newcommand{\bfMeightICTatFifty}{-0.00}
\newcommand{\bfMeightICTatThreeH}{-0.00}

\newcommand{\bfMeightICTtime}{0.18}

\newcommand{\bfMeightICTvsFPmean}{+38.20}
\newcommand{\bfMeightICTvsFPwins}{100.0}
\newcommand{\bfMeightICTvsFPlosses}{0.0}

\newcommand{\bfMeightICTvsXMmean}{+0.13}
\newcommand{\bfMeightICTvsXMwins}{5.0}
\newcommand{\bfMeightICTvsXMlosses}{0.0}

\newcommand{\bfMeightICTatFiftyVsFPfinal}{+38.20}

\newcommand{\bfMfourWN}{40}

\newcommand{\bfMfourWFPitNNN}{1415}

\newcommand{\bfMfourWXMitNNN}{16}

\newcommand{\bfMfourWICTitNNN}{12}

\newcommand{\bfMfourWICTvsFPmean}{+18.72}

\newcommand{\bfMfourWICTvsXMmean}{+0.27}

\newcommand{\slqpKAfiftysixSINR}{-17.9}

\newcommand{\slqpKAfiftysixQFTitNNN}{6}

\newcommand{\slqpKAfiftysixQFTwithin}{100.0}

\newcommand{\slqpKAfiftysixLFTrelBest}{-64.70}

\newcommand{\slqpKAfiftysixFPitNNN}{6}

\newcommand{\slqpKAfiftysixICTitNNN}{4}

\newcommand{\slqpKAfiftysixICTwithin}{100.0}

\newcommand{\slqpKAfiftysixICTtotal}{0.19}
\newcommand{\slqpKAfiftysixICTNewton}{29.9}

\newcommand{\slqpKAfiftysixICTvsQFTtimeRatio}{1.8}

\newcommand{\slqpKAseventyQFTitNNN}{5}

\newcommand{\slqpKAseventyQFTwithin}{100.0}

\newcommand{\slqpKAseventyLFTrelBest}{-83.92}

\newcommand{\slqpKAseventyICTitNNN}{5}

\newcommand{\slqpKAseventyICTwithin}{95.0}

\newcommand{\slqpKAseventyICTvsQFTtimeRatio}{2.1}

\newcommand{\slqpSYSINR}{-10.2}

\newcommand{\slqpSYFPitNNN}{12}

\newcommand{\slqpSYICTitNNN}{8}

\newcommand{\slqpSYN}{100}
\newcommand{\slqpSYICTvsQFTmean}{+1.41}
\newcommand{\slqpSYICTvsQFTwins}{13.0}
\newcommand{\slqpSYICTvsQFTlosses}{6.0}

\newcommand{\slqpSYICTvsSEFPwins}{8.0}
\newcommand{\slqpSYICTvsSEFPlosses}{3.0}

\newcommand{\slqpTableRows}{\multirow{5}{*}{$K=56$, $K_q=14$} & QFT \cite{Khan24} & Clarabel & 4 & 6 & -0.00\% & 100.0\% & 34 & 0.38\\
 & LFT \cite{Khan24} & Clarabel & 195 & 200 & -64.70\% & 0.0\% & 42 & 8.39\\
 & LDT+QT, dual & Newton & 6 & 6 & -0.00\% & 100.0\% & 18 & 0.23\\
 & RIT+QT, dual & Newton & 4 & 4 & -0.00\% & 100.0\% & 17 & 0.17\\
 & \textbf{ICT, dual (Alg.~\ref{alg:slqp})} & Newton & 4 & 4 & -0.00\% & 100.0\% & 19 & 0.19\\
\midrule
\multirow{5}{*}{$K=70$, $K_q=7$} & QFT \cite{Khan24} & Clarabel & 4 & 5 & -0.00\% & 100.0\% & 51 & 0.46\\
 & LFT \cite{Khan24} & Clarabel & 172 & 198 & -83.92\% & 0.0\% & 134 & 24.62\\
 & LDT+QT, dual & Newton & 6 & 6 & -0.00\% & 100.0\% & 21 & 0.32\\
 & RIT+QT, dual & Newton & 4 & 5 & -0.00\% & 100.0\% & 21 & 0.18\\
 & \textbf{ICT, dual (Alg.~\ref{alg:slqp})} & Newton & 4 & 5 & -0.01\% & 95.0\% & 24 & 0.22\\
\midrule
\multirow{5}{*}{$K=28$, $K_q=7$ (interf.-lim.)} & QFT \cite{Khan24} & Clarabel & 6 & 8 & -2.59\% & 79.0\% & 21 & 0.35\\
 & LFT \cite{Khan24} & Clarabel & 198 & 200 & -61.98\% & 1.0\% & 18 & 3.59\\
 & LDT+QT, dual & Newton & 10 & 12 & -2.32\% & 80.0\% & 13 & 0.26\\
 & RIT+QT, dual & Newton & 6 & 7 & -2.45\% & 82.0\% & 13 & 0.19\\
 & \textbf{ICT, dual (Alg.~\ref{alg:slqp})} & Newton & 6 & 8 & -2.06\% & 85.0\% & 15 & 0.23\\}
\newcommand{\slqpN}{20}
\newcommand{\slqpAscentN}{140}
\newcommand{\slqpAscentSEFP}{0.10}
\newcommand{\slqpAscentFP}{8.25}

\newcommand{\scalingTableRows}{56 & 14 & 25 & 21 & 20 & 5854\\
70 & 7 & 26 & 25 & 45$^\dagger$ & 436\\
140 & 14 & 60 & 57 & 744 & 2160\\
280 & 28 & 196 & 228 & 2866$^\dagger$ & --\\}
\newcommand{\scalingRatioOneForty}{13}
\newcommand{\scalingRatioTwoEighty}{13}
\newcommand{\scalingScsRatioFiftySix}{283}
\newcommand{\scalingScsRatioOneForty}{38}

\newtheorem{theorem}{Theorem}
\newtheorem{proposition}{Proposition}
\newtheorem{corollary}{Corollary}

\newtheorem{remark}{Remark}

\newcommand{\R}{\mathbb{R}}
\newcommand{\C}{\mathbb{C}}
\newcommand{\Om}{\Omega}
\newcommand{\rhof}{\varrho}
\newcommand{\LQ}{\mathrm{LQ}}
\newcommand{\RQ}{\mathrm{RQ}}
\newcommand{\ICT}{\mathrm{ICT}}
\newcommand{\tr}{\operatorname{tr}}
\newcommand{\diag}{\operatorname{diag}}
\newcommand{\Real}{\operatorname{Re}}

\newcommand{\mbps}{\,\mathrm{Mbps}}
\DeclareMathOperator*{\argmax}{arg\,max}

\begin{document}

\title{The Interference-Conjugate Transform: The Tightest Interference-Linear Surrogate for Power Control, Scheduling and Beamforming in Wireless Networks}

\author{Ahmad~Khan%
\thanks{Manuscript draft, October 2026. The author is with Ericsson Ottawa, 349 Terry Fox Drive, Kanata, ON K2K 2V6, Canada.}}

\markboth{Draft}{Khan: The Interference-Conjugate Transform}

\makeatletter
\def\@IEEENORMtitlevspace{3.0\baselineskip}\def\@IEEEMINtitlevspace{1\baselineskip}
\makeatother
\maketitle

\begin{abstract}
Closed-form fractional programming (FP), the weighted minimum mean-squared error (WMMSE) algorithm, and the two surrogate-enhanced methods proposed in 2026---the surrogate-enhanced FP (SEFP) of Jiao \emph{et al.} and the XMMSE algorithm of Yang and Zhao---all maximize the weighted sum rate (WSR) by iteratively maximizing a surrogate function that is \emph{linear in the interference power} of each link. In this paper, we show that this family admits a unique best member. We propose the interference-conjugate transform (ICT), which replaces each rate term $\log(1+a/b)$ by $\max_{\kappa\ge0}[\Omega(\kappa a)-\kappa b]$, where $\Omega$ is the closed-form partial convex conjugate of the rate in the interference variable and the optimal multiplier is the classical interference price $a/(b(a+b))$. We prove that, for any fixed price, $\Omega(\kappa\,\cdot)$ is the pointwise largest own-signal function compatible with a valid interference-linear lower bound, and that tightness at an anchor point forces the price. Consequently, all four benchmark surrogates, which we show share this price and differ only in their own-signal term, are dominated pointwise by the ICT surrogate; as a by-product, we prove that the single-stream XMMSE surrogate coincides exactly with that of SEFP. At a signal-to-interference-plus-noise ratio $\gamma$, the own-signal curvature of the benchmark surrogates exceeds the true curvature by factors that behave as $\gamma/2$ (FP/WMMSE) and $(\log\gamma)/2$ (SEFP/XMMSE), against $1+1/(\gamma(\gamma+2))$ for ICT, and this ordering carries over to the certified ascent per iteration and to the asymptotic linear convergence rate. The ICT yields closed-form updates of the same per-iteration cost as the benchmarks for single-antenna power control, uplink joint scheduling and power control, and multi-cell MISO beamforming, and it yields the first solver-free algorithm for sum-least-$q$th-percentile rate maximization, whose subproblem is solved through a $K$-dimensional convex dual with closed-form gradient and Hessian and a certified duality gap. In the uplink scheduling setting of Jiao \emph{et al.}, ICT improves the proportional-fairness utility over FP by \pfICTminusFP{} and over SEFP by \pfICTminusSEFP{} (\pfICTminusSEFPwins{}/\pfNdrops{} drops) with \pfItersICT{} iterations per slot instead of \pfItersFP{} and \pfItersSEFP{}; in MISO beamforming with $M=8$ antennas, it reaches $99.9\%$ of its final value in a median of \bfMeightICTitNNN{} iterations against \bfMeightXMitNNN{} for XMMSE and \bfMeightFPitNNN{} for FP/WMMSE.
\end{abstract}


\section{Introduction}
\IEEEPARstart{W}{eighted} sum-rate (WSR) maximization under the treatment of interference as noise is the mainstay formulation of radio resource management in interference-limited wireless networks. Power control, user scheduling and beamforming are all instances of
\begin{equation}
\max_{\mathbf x\in\mathcal X}\; \sum_{i} w_i \log\Big(1+\frac{a_i(\mathbf x)}{b_i(\mathbf x)}\Big),
\label{eq:wsr}
\end{equation}
in which $a_i$ denotes the useful signal power of link $i$ and $b_i$ its interference-plus-noise power. Problem \eqref{eq:wsr} is non-convex and NP-hard in general \cite{LuoZhang08}, and the practical state-of-the-art consists of iterative algorithms that are guaranteed to converge to stationary points. Three ingredients have dominated this class of problems for more than a decade: the WMMSE algorithm \cite{Christensen08,Shi11}, interference pricing \cite{Huang06,ShiBerryHonig09}, and the fractional programming (FP) framework of Shen and Yu \cite{ShenYuI,ShenYuII}, whose closed-form variant composes a Lagrangian dual transform (LDT) with a quadratic transform (QT). It was later recognized \cite{ShenMFP19,Zhang23,Shen24} that all of these are instances of minorization--maximization (MM) \cite{Sun17}: each iteration replaces \eqref{eq:wsr} by a surrogate function that lies below the objective, touches it at the current point, and is easy to maximize; furthermore, closed-form FP and WMMSE generate identical iterates.

Two recent papers improve the quality of the surrogate rather than the algorithmic wrapper around it. Jiao \emph{et al.} \cite{Jiao26} observe that the LDT is a tangent bound in the reciprocal coordinate $z=1/(1+r)$ of the SINR $r$ which makes it conservative, and propose the reciprocal-inversion transform (RIT), a tangent bound of the concave function $1/\log(1+1/\tau)$ in the coordinate $\tau=1/r$; combined with the QT, it yields the surrogate-enhanced FP (SEFP) algorithm, which retains the per-cell separability and closed-form updates of FP and consistently outperforms it in uplink joint scheduling and power control. Independently, Yang and Zhao \cite{YangZhao26} start from the exact integral relation $\log\det(\mathbf I+\mathbf X)=\int_0^1\tr[(\mathbf I+\varepsilon\mathbf X)^{-1}\mathbf X]\,d\varepsilon$ between the rate and a family of regularized mean-squared errors, and derive the XMMSE algorithm, whose quadratic surrogate is provably tighter than that of WMMSE and converges faster, in particular at high SNR.

Both papers improve one specific bound and compare it with one specific benchmark. In this paper, we address the obvious question that these results invite: \emph{among all surrogates with the structure that makes the updates separable and closed form, which one is the best, and how far from it are FP, WMMSE, SEFP and XMMSE?} The structure in question is \emph{linearity in the interference power}. Interference is the only quantity that couples the links; hence, if a surrogate charges a fixed price per unit of interference, the maximization decouples into per-link subproblems, which is precisely why closed-form FP, SEFP, XMMSE and the pricing algorithms are computationally cheap. We show that this family contains a unique tightest surrogate, that it admits a closed form, that the four benchmarks are strictly dominated members of the family, and that the algorithms it yields have the same per-iteration cost as the benchmarks.

\subsection{Contributions}
The contributions of this paper can be summarized as follows.
\begin{itemize}
\item \textbf{The interference-conjugate transform (ICT).} For $a\ge0$ and $b>0$, we establish the identity
\begin{equation*}
\log\Big(1+\frac ab\Big)=\max_{\kappa\ge0}\big[\Om(\kappa a)-\kappa b\big],\qquad \kappa^\star=\frac{a}{b(a+b)},
\end{equation*}
with $\Om(u)=\log(1+\rhof(u))+\rhof(u)/(1+\rhof(u))$ and $\rhof(u)=\tfrac12(u+\sqrt{u^2+4u})$. $\Om$ is the partial convex conjugate of the rate in the interference variable alone, and $\kappa^\star$ is the classical interference price \cite{Huang06}. For fixed $\kappa$ the transformed objective is concave and separable in the powers, and the transformed problem is block-concave, so that alternating maximization is a block-coordinate ascent whose fixed points are the KKT points of \eqref{eq:wsr}.
\item \textbf{Tightness (Theorem~\ref{thm:tightest}).} For any fixed price $\kappa$, every function $\phi$ such that $\phi(a)-\kappa b\le\log(1+a/b)$ for all $(a,b)$ satisfies $\phi\le\Om(\kappa\,\cdot)$, and a surrogate of this form that is tight at an anchor $(a_0,b_0)$ must use $\kappa=\kappa^\star(a_0,b_0)$. Hence, ICT is the pointwise largest interference-linear minorizer, and the interference price is a consequence of tightness rather than a design choice.
\item \textbf{One family, one price, four dominated members (Propositions~\ref{prop:family} and \ref{prop:xmmse}).} We show that LDT+QT (closed-form FP), WMMSE, RIT+QT (SEFP) and single-stream XMMSE are interference-linear surrogates with the common price $w_i\kappa_i^\star$, whose own-signal functions in the normalized variable $\nu=a/a_0$ are
\begin{align*}
\phi_{\LQ}(\nu)&=\log(1+\gamma)-\gamma+2\gamma\sqrt\nu-\tfrac{\gamma^2}{1+\gamma}\nu,\\
\phi_{\RQ}(\nu)&=2\log(1+\gamma)\sqrt\nu-\tfrac{b(\gamma)}{1+\gamma}\nu,\\
\phi_{\ICT}(\nu)&=\Om\big(\tfrac{\gamma^2}{1+\gamma}\nu\big),
\end{align*}
with $b(\gamma)=(1+\gamma)\log(1+\gamma)-\gamma$, and that $\phi_{\LQ}\le\phi_{\RQ}\le\phi_{\ICT}\le\log(1+\gamma\nu)+\gamma/(1+\gamma)$ with equality only at $\nu=1$. Crucially, we prove that the single-stream XMMSE surrogate equals $\phi_{\RQ}$ exactly, including its complex phase structure for beamforming; hence XMMSE and SEFP are the same method for single-stream users.
\item \textbf{Quantitative consequences (Propositions~\ref{prop:curv}--\ref{prop:rate}).} The own-signal curvature at the anchor exceeds the true curvature by the factors $(1+\gamma)^2/(2\gamma)$, $(1+\gamma)^2\log(1+\gamma)/(2\gamma^2)$ and $1+1/(\gamma(\gamma+2))$, respectively, and the certified ascent of one MM step and the asymptotic linear rate at a non-degenerate stationary point are ordered accordingly.
\item \textbf{Algorithms.} We derive closed-form ICT updates for single-antenna power control, for uplink joint scheduling and power control (where the candidate score becomes $w_k\log(1+\rho_k)$, the exact best response to the current prices, in place of the bounded LDT score of \cite{ShenYuII}), and for multi-cell MISO beamforming, where a rank-one price and a Cauchy--Schwarz argument reduce the per-base-station problem to an eigendecomposition and a scalar bisection. All of these have the per-iteration complexity of their FP counterparts.
\item \textbf{A solver-free algorithm for percentile optimization (Section~\ref{sec:slqp}).} For the sum-least-$q$th-percentile (SLqP) rate problem of Khan and Adve \cite{Khan24}, whose QFT and LFT subproblems require a conic solver and for which closed-form algorithms were stated to be unknown, we show that the ICT subproblem admits a $K$-dimensional convex dual with closed-form gradient and Hessian and a certified duality gap. To the best of our knowledge, this is the first solver-free algorithm for this class of problems; it reaches the same values as QFT in fewer outer iterations, and its certified subproblem solve is an order of magnitude faster than an interior-point conic solver from about 140 links upward.
\item \textbf{Simulations.} In the settings of \cite{ShenYuI,ShenYuII,Jiao26,Khan24}, ICT reaches equal or better stationary points in one to two orders of magnitude fewer iterations than FP and in fewer iterations than SEFP/XMMSE, and under proportional-fairness scheduling its utility gain over SEFP is as large as the gain of SEFP over FP.
\end{itemize}

\subsection{Related Work, Organization and Notation}
The MM interpretation of FP and WMMSE is due to \cite{ShenMFP19,Zhang23}, and \cite{Shen24} accelerates QT/WMMSE by extrapolation, which is compatible with any choice of surrogate. Interference pricing \cite{Huang06,ShiBerryHonig09} uses the price $\kappa^\star$ with a best-response own-signal term, which is not a valid global minorizer, so that monotonicity requires additional conditions; ICT is the tightest \emph{valid} choice. SCALE \cite{Papandriopoulos09} uses the bound $\log(1+z)\ge\alpha\log z+\beta$, which is linear neither in the interference nor in the signal. The QFT and LFT of Khan and Adve \cite{Khan24} are logarithms of affine functions of the interference and require a conic solver per iteration; they lie outside the family studied here and are not pointwise comparable with ICT, but Section~\ref{sec:slqp} compares with them on their own SLqP problem. We confine ourselves to single-stream users; the multi-stream $\log\det$ case is discussed in Section~\ref{sec:concl}.

Section~\ref{sec:prelim} formulates the problem and places the four benchmark surrogates in a common notation; Section~\ref{sec:ict} introduces the ICT, proves its tightness and derives its consequences; Sections~\ref{sec:siso}, \ref{sec:miso} and \ref{sec:slqp} apply it to SISO power control and scheduling, MISO beamforming and percentile optimization, each with a proposed approach followed by simulation results; Section~\ref{sec:concl} concludes. Throughout, $\log$ denotes the natural logarithm; rates in figures are in Mbps over the system bandwidth; $\gamma=a_0/b_0$ denotes the SINR at the anchor point; and for a Hermitian matrix $\mathbf A$, $\mathbf A\succeq0$ means positive semidefinite.

\section{Preliminaries}
\label{sec:prelim}
\subsection{Problem Formulation}
We consider \eqref{eq:wsr} with $a_i,b_i$ determined by the decision variables, in three instances. \emph{SISO power control:} $K$ links, $a_i=g_{ii}p_i$, $b_i=\sum_{j\ne i}g_{ij}p_j+\sigma^2$, $p_{\min}\le p_i\le P_{\max}$. \emph{Uplink joint scheduling and power control} \cite{ShenYuII,Jiao26}: cells $i\in\mathcal B$, each scheduling one user $s_i\in\mathcal K_i\cup\{\varnothing\}$ at power $p_{s_i}$; $a_i=g_{i,s_i}p_{s_i}$, $b_i=\sum_{j\ne i}g_{i,s_j}p_{s_j}+\sigma^2$, with $g_{i,k}$ the uplink gain of user $k$ at base station (BS) $i$ and weight $w_{s_i}$. \emph{Multi-cell MISO beamforming} \cite{ShenYuI}: BS $j$ with $M$ antennas serves user $j$ with $\mathbf v_j\in\C^M$, $\|\mathbf v_j\|^2\le P$; $a_i=|\mathbf h_{ii}^H\mathbf v_i|^2$, $b_i=\sum_{j\ne i}|\mathbf h_{ij}^H\mathbf v_j|^2+\sigma^2$, with $\mathbf h_{ij}$ the channel from BS $j$ to user $i$.

An MM algorithm \cite{Sun17} for \eqref{eq:wsr} chooses, at the current point $\bar{\mathbf x}$, a surrogate $S(\mathbf x;\bar{\mathbf x})\le f(\mathbf x)$ with $S(\bar{\mathbf x};\bar{\mathbf x})=f(\bar{\mathbf x})$ and moves to a maximizer of $S$. The objective is then non-decreasing, and the limit points are stationary under mild conditions. For WSR problems every surrogate of interest is built term by term from a scalar minorizer of $\log(1+a/b)$ at the anchor $(a_0,b_0)$, which we refer to as the \emph{single-ratio surrogate}.

\subsection{The Four Benchmark Surrogates}
\label{sec:benchmarks}
\emph{LDT+QT (closed-form FP).} The LDT \cite{ShenYuI} is the tangent of $-\log z$ at $z_0=1/(1+\gamma)$ in the coordinate $z=1/(1+r)$, $r=a/b$:
\begin{equation}
\log(1+r)\ \ge\ \ell_{\mathrm{LDT}}(r;\gamma)\triangleq\log(1+\gamma)-\gamma+\frac{(1+\gamma)r}{1+r}.
\label{eq:ldt}
\end{equation}
The ratio $r/(1+r)=a/(a+b)$ is then decoupled by the QT \cite{ShenYuI}, $A/B=\max_y[2y\sqrt A-y^2B]$ with $y^\star=\sqrt A/B$. With $y$ fixed at its anchor value, the resulting single-ratio surrogate is
\begin{equation}
\ell^{\LQ}(a,b)=\log(1+\gamma)-\gamma+\frac{2(1+\gamma)\sqrt{a_0a}}{a_0+b_0}-\frac{(1+\gamma)a_0}{(a_0+b_0)^2}(a+b).
\label{eq:lq}
\end{equation}

\emph{WMMSE.} With the MMSE receiver $u_0=\sqrt{a_0}/(a_0+b_0)$, the MSE $e(a,b)=1-2u_0\sqrt a+u_0^2(a+b)$ and the weight $\tilde w=1/e(a_0,b_0)=1+\gamma$, the WMMSE surrogate \cite{Shi11,Zhang23} is $\log\tilde w-\tilde w\,e(a,b)+1$, which equals \eqref{eq:lq} identically (this is the scalar form of the equivalence of closed-form FP and WMMSE established in \cite{Zhang23}).

\emph{RIT+QT (SEFP).} Jiao \emph{et al.} \cite{Jiao26} define $b(\gamma)=(1+\gamma)\log(1+\gamma)-\gamma$ and $c(\gamma)=(1+\gamma)\log^2(1+\gamma)$ and prove that
\begin{equation}
\log(1+r)\ \ge\ \ell_{\mathrm{RIT}}(r;\gamma)\triangleq\frac{c(\gamma)\,r}{\gamma^2+b(\gamma)\,r},
\label{eq:rit}
\end{equation}
a tangent bound of the concave function $\tau\mapsto1/\log(1+1/\tau)$ at $\tau=1/\gamma$. Applying the QT to the ratio in \eqref{eq:rit} with $y$ fixed at its anchor value gives
\begin{equation}
\ell^{\RQ}(a,b)=\frac{2\log(1+\gamma)}{\sqrt{a_0}}\sqrt a-\frac{b(\gamma)}{(1+\gamma)a_0}a-\frac{\gamma}{(1+\gamma)b_0}\,b,
\label{eq:rq}
\end{equation}
where we have used $y^\star=\sqrt{c(\gamma)a_0}/(\gamma^2b_0+b(\gamma)a_0)=1/\sqrt{(1+\gamma)a_0}$.

\emph{XMMSE.} For a single stream, the sum-regularized-MSE surrogate of \cite{YangZhao26} is a quadratic in the complex scalar $z=\mathbf h^H\mathbf v$ (Section~\ref{sec:miso}); its scalar form at the anchor $x_0=\sqrt{a_0}$, $b=b_0$, is the quadratic in the amplitude $x=\sqrt a$ given in \cite[Sec.~IV]{YangZhao26},
\begin{multline}
g_X(x)=\Big[\frac1{1+\gamma}-\frac{\log(1+\gamma)}{\gamma}\Big](x-x_0)^2\\
+\frac{2x_0}{1+\gamma}(x-x_0)+\log(1+\gamma),
\label{eq:xmmse-scalar}
\end{multline}
to be compared with the WMMSE quadratic $g_W$, which has $1/(1+\gamma)-1$ in place of the bracket.

It is critical to recognize that all four surrogates have the form $\phi(a)-\lambda b$: they are \emph{affine in the interference}, with a price $\lambda$ that does not depend on $a$. This is the structural property that makes the per-link updates separable and closed form, and it is the family over which we optimize in the sequel.

\section{The Interference-Conjugate Transform}
\label{sec:ict}
\subsection{The Transform}
\begin{theorem}[Interference-conjugate transform]
\label{thm:ict}
Let $\rhof(u)=\tfrac12\big(u+\sqrt{u^2+4u}\big)$ for $u\ge0$, the unique non-negative root of $\rhof^2/(1+\rhof)=u$, and define
\begin{equation}
\Om(u)=\log\big(1+\rhof(u)\big)+\frac{\rhof(u)}{1+\rhof(u)}.
\label{eq:omega}
\end{equation}
For every $a\ge0$ and $b>0$,
\begin{equation}
\log\Big(1+\frac ab\Big)=\max_{\kappa\ge0}\big[\Om(\kappa a)-\kappa b\big],\qquad \kappa^\star=\frac{a}{b(a+b)},
\label{eq:ict}
\end{equation}
and the maximizer is unique when $a>0$. Consequently, $\Om(\kappa a)-\kappa b\le\log(1+a/b)$ for every $\kappa\ge0$.
\end{theorem}
\begin{IEEEproof}
Fix $a>0$ and $\kappa>0$, and consider $g(b)=\log(1+a/b)+\kappa b$ on $b>0$. Since $\log(1+a/b)=\log(a+b)-\log b$ has second derivative $1/b^2-1/(a+b)^2>0$, $g$ is strictly convex; it tends to $+\infty$ as $b\to0^+$ and as $b\to\infty$, so it has a unique minimizer $b^\star$, characterized by $g'(b^\star)=0$, i.e., $\kappa=a/(b^\star(a+b^\star))$. Writing $\rho\triangleq a/b^\star$, this condition reads $\kappa a=\rho^2/(1+\rho)$, hence $\rho=\rhof(\kappa a)$, and the minimum value is $\log(1+\rho)+\kappa b^\star=\log(1+\rho)+\rho/(1+\rho)=\Om(\kappa a)$. Hence $\Om(\kappa a)\le\log(1+a/b)+\kappa b$ for all $b>0$, with equality if and only if $b=b^\star$, i.e., if and only if $\kappa=a/(b(a+b))=\kappa^\star$. This proves the minorization and the identity \eqref{eq:ict} for $a>0$ (for $\kappa=0$ the inequality is trivial, and for $a=0$ both sides of \eqref{eq:ict} vanish at $\kappa=0$). Uniqueness of the maximizer follows from the strict concavity of $\kappa\mapsto\Om(\kappa a)$ established in Proposition~\ref{prop:omega}(a) below.
\end{IEEEproof}

The identity \eqref{eq:ict} is the Fenchel--Legendre duality of the rate in the interference variable alone: $b\mapsto\log(1+a/b)$ is convex, and $\Om(\kappa a)=\inf_b[\log(1+a/b)+\kappa b]$ is the intercept of its tightest affine minorant of slope $-\kappa$. We emphasize that the own-signal variable $a$ is never approximated.

\begin{proposition}[Properties of $\Om$]
\label{prop:omega}
(a) $\Om$ is strictly increasing and strictly concave on $[0,\infty)$, $\Om(0)=0$, $\Om'(u)=1/\rhof(u)$ and $\Om''(u)=-(1+\rhof)^2/(\rhof^3(\rhof+2))$.
(b) $\Om(u)=2\sqrt u-\tfrac u2+\tfrac{u^{3/2}}{12}+O(u^2)$ as $u\to0$, and $\Om(u)=\log u+1+\tfrac1u+O(u^{-2})$ as $u\to\infty$.
(c) $(a,\kappa)\mapsto\Om(\kappa a)$ is jointly concave on $\R_+^2$.
(d) At the maximizer of \eqref{eq:ict}, $\rhof(\kappa^\star a)=a/b$, i.e., $\rhof$ returns the SINR.
\end{proposition}
\begin{IEEEproof}
(a) Taking $a=1$ in the proof of Theorem~\ref{thm:ict}, $\Om(u)=\inf_{b>0}[\log(1+1/b)+ub]$ with minimizer $b^\star(u)=1/\rhof(u)$. By the envelope theorem, $\Om'(u)=b^\star(u)=1/\rhof(u)>0$; since $\rhof$ is strictly increasing, $\Om$ is strictly concave, and $\Om''=-\rhof'/\rhof^2$ with $\rhof'=(1+\rhof)^2/(\rhof(\rhof+2))$ obtained by implicit differentiation of $\rhof^2/(1+\rhof)=u$. (b) Substituting $\rhof(u)=\sqrt u+u/2+u^{3/2}/8+O(u^2)$ and $\rhof(u)=u+1-1/u+O(u^{-2})$ into \eqref{eq:omega} yields the expansions. (c) For each fixed $b>0$, $(a,\kappa)\mapsto\log(1+a/b)+\kappa b$ is concave (concave in $a$ and affine in $\kappa$), and an infimum of concave functions is concave. (d) was shown in the proof of Theorem~\ref{thm:ict}, where $\rho=a/b^\star$ and $b^\star=b$ at $\kappa=\kappa^\star$.
\end{IEEEproof}

Part (b) exhibits the two regimes of the ICT: at low SINR, $\Om(\kappa a)\approx2\sqrt{\kappa a}$ is the QT-type square-root bound, so that ICT and FP agree to first order; at high SINR, $\Om(\kappa a)\approx\log(\kappa a)+1$ has the logarithmic growth of the rate itself. Part (c) makes the transformed objective $\sum_iw_i[\Om(\kappa_ia_i)-\kappa_ib_i]$ jointly concave in $(\mathbf a,\boldsymbol\kappa)$ and, for fixed $\boldsymbol\kappa$, concave and separable in the powers.

\subsection{Derivation and Relation to Existing Transforms}
\label{sec:derivation}
The construction has three steps, each of which is forced once the goal is stated. (i) \emph{What must be approximated.} The links of \eqref{eq:wsr} are coupled only through $b_i$; for fixed $b_i$ the rate is already concave in the own signal. Hence, the interference is the only variable that needs replacing, and approximating the own signal as well---as the QT does when it conjugates the whole denominator $a+b$ of the LDT ratio $a/(a+b)$---gives away accuracy that was never needed. (ii) \emph{The simpler form.} A function of the sum $\sum_{j\ne i}g_{ij}p_j$ splits into per-link terms only if it is affine, with a slope independent of $a_i$ (otherwise bilinear terms $p_ip_j$ appear); this leads to a fixed price $\kappa_i$ per unit of interference. (iii) \emph{The tightest such form.} Since $b\mapsto\log(1+a/b)$ is convex, the tightest affine minorant of slope $-\kappa$ that is valid for \emph{every} $b$ has intercept $\inf_b[\log(1+a/b)+\kappa b]$, the partial Fenchel conjugate, and \eqref{eq:omega} is the result; the dual variable is an interference price by construction.

Three remarks place the ICT among the known transforms. First, it is the same operation as the QT, which is the partial conjugate of the ratio $A/B$ in $B$; the QT conjugates a ratio after the logarithm has been removed by a separate bound (LDT or RIT), whereas the ICT conjugates the logarithm of the ratio in one step, and one bound loses less than two composed bounds (Proposition~\ref{prop:family}). Second, the closed form rests on the fact that the stationarity condition in $b$ is a quadratic equation (the origin of $\sqrt{u^2+4u}$), because the rate is the logarithm of a ratio of affine functions; for multi-stream $\log\det$ rates the analogous equation is of Riccati type and does not collapse (Remark~\ref{rem:mimo}). Third, the ICT is distinct from Dinkelbach's and the Charnes--Cooper transforms, from the LDT and the RIT (tangents in $1/(1+r)$ and $1/r$), from SCALE's bound, and from the QFT/LFT of \cite{Khan24}, which lie outside the family; by Proposition~\ref{prop:omega}(b), its own-signal term is the QT square root at low SINR and the logarithm at high SINR.

\subsection{Tightness}
\begin{theorem}[Tightest interference-linear minorizer]
\label{thm:tightest}
Fix $\kappa\ge0$.
(i) If $\phi:\R_+\to\R$ satisfies $\phi(a)-\kappa b\le\log(1+a/b)$ for all $a\ge0$, $b>0$, then $\phi(a)\le\Om(\kappa a)$ for all $a\ge0$.
(ii) If in addition $\phi(a_0)-\kappa b_0=\log(1+a_0/b_0)$ for some $a_0>0$, $b_0>0$, then $\kappa=\kappa^\star(a_0,b_0)$ and $\phi(a_0)=\Om(\kappa a_0)$.
\end{theorem}
\begin{IEEEproof}
(i) For fixed $a$, $\phi(a)\le\inf_{b>0}[\log(1+a/b)+\kappa b]$, and this infimum equals $\Om(\kappa a)$ by the proof of Theorem~\ref{thm:ict}. (ii) From (i) and the hypothesis, $\Om(\kappa a_0)-\kappa b_0\ge\phi(a_0)-\kappa b_0=\log(1+a_0/b_0)=\max_{\kappa'}[\Om(\kappa'a_0)-\kappa'b_0]$, so that $\kappa$ attains the maximum in \eqref{eq:ict}, which is unique for $a_0>0$ by the strict concavity of $\kappa\mapsto\Om(\kappa a_0)$ (Proposition~\ref{prop:omega}(a)). Hence $\kappa=\kappa^\star$, and the two inequalities are equalities.
\end{IEEEproof}

Theorem~\ref{thm:tightest} says two things. First, the interference price of every tight interference-linear surrogate is forced: it must be the marginal rate loss per unit of interference, $-\partial\log(1+a/b)/\partial b=a/(b(a+b))$, the price of \cite{Huang06,ShiBerryHonig09}; that WMMSE, FP, SEFP and XMMSE, derived by unrelated routes, all charge this price is therefore explained rather than coincidental. Second, once the price is fixed, the own-signal function is bounded above by $\Om(\kappa\,\cdot)$: everything a surrogate loses relative to ICT is lost on the own-signal side. The proof is short because the construction \emph{is} the optimality argument; the content lies in the closed form \eqref{eq:omega} and in the next proposition, which places the benchmarks in the family.

\subsection{One Family, One Price}
\begin{proposition}[Common structure of the benchmarks]
\label{prop:family}
At an anchor $(a_0,b_0)$ with $\gamma=a_0/b_0$, let $\nu=a/a_0$. The surrogates \eqref{eq:lq}, \eqref{eq:rq} and the ICT surrogate $\Om(\kappa^\star a)-\kappa^\star b$ can be written as $\phi_\bullet(\nu)-\kappa^\star b$ with the \emph{same} price $\kappa^\star=\gamma/((1+\gamma)b_0)$ and the own-signal functions
\begin{align}
\phi_{\LQ}(\nu)&=\log(1+\gamma)-\gamma+2\gamma\sqrt\nu-\frac{\gamma^2}{1+\gamma}\,\nu,\label{eq:philq}\\
\phi_{\RQ}(\nu)&=2\log(1+\gamma)\sqrt\nu-\frac{b(\gamma)}{1+\gamma}\,\nu,\label{eq:phirq}\\
\phi_{\ICT}(\nu)&=\Om\Big(\frac{\gamma^2}{1+\gamma}\,\nu\Big).\label{eq:phiict}
\end{align}
All three equal $\log(1+\gamma)+\gamma/(1+\gamma)=\Om(\kappa^\star a_0)$ at $\nu=1$ with derivative $\gamma/(1+\gamma)$, and
\begin{equation}
\phi_{\LQ}\le\phi_{\RQ}\le\phi_{\ICT}\le\log(1+\gamma\nu)+\frac{\gamma}{1+\gamma}
\label{eq:chain}
\end{equation}
on $\nu\ge0$, with equalities only at $\nu=1$. Explicitly, $\phi_{\RQ}(\nu)-\phi_{\LQ}(\nu)=(\gamma-\log(1+\gamma))(\sqrt\nu-1)^2$.
\end{proposition}
\begin{IEEEproof}
Substituting $y^\star$ into \eqref{eq:lq}, we have $2y^\star\sqrt{(1+\gamma)a}=2\gamma\sqrt\nu$ because $a_0/(a_0+b_0)=\gamma/(1+\gamma)$, and $(y^\star)^2=(1+\gamma)a_0/(a_0+b_0)^2=\kappa^\star$, so that the coefficient of $b$ is $-\kappa^\star$ and that of $a$ is $-\kappa^\star a_0\nu=-\gamma^2\nu/(1+\gamma)$. For \eqref{eq:rq}, $y^\star=1/\sqrt{(1+\gamma)a_0}$ gives $2y^\star\sqrt{c(\gamma)a}=2\log(1+\gamma)\sqrt\nu$, the coefficient of $a$ is $-(y^\star)^2b(\gamma)$, and the price is $(y^\star)^2\gamma^2=\gamma^2/((1+\gamma)a_0)=\kappa^\star$. The ICT surrogate is $\Om(\kappa^\star a_0\nu)-\kappa^\star b$ with $\kappa^\star a_0=\gamma^2/(1+\gamma)$. The difference $\phi_{\RQ}-\phi_{\LQ}$ is a direct computation (it is Proposition~3 of \cite{Jiao26}); $\phi_{\RQ}\le\phi_{\ICT}$ is Theorem~\ref{thm:tightest}(i) applied to the valid minorizer $\ell^{\RQ}$; and the last inequality is Theorem~\ref{thm:ict} with $b=b_0$. Strictness follows from the strict concavity of $\kappa\mapsto\Om(\kappa a)$ for the last two inequalities and from the explicit square for the first.
\end{IEEEproof}

\begin{proposition}[XMMSE coincides with RIT+QT for one stream]
\label{prop:xmmse}
For a single-stream user with scalar effective channel $z=\mathbf h^H\mathbf v$, interference-plus-noise $b$ and anchor $(z_0,b_0)$, $\gamma=|z_0|^2/b_0$, the XMMSE surrogate of \cite{YangZhao26} is
\begin{equation}
\frac{2\log(1+\gamma)}{\gamma\,b_0}\,\Real(z_0^*z)-\frac{b(\gamma)}{\gamma(1+\gamma)b_0}\,|z|^2-\frac{\gamma}{(1+\gamma)b_0}\,b,
\label{eq:xmmse-vec}
\end{equation}
and so is the RIT+QT surrogate obtained by applying the complex QT to \eqref{eq:rit}. In particular, the scalar forms satisfy $g_X(x)\equiv\ell^{\RQ}(x^2,b_0)$ and $g_W(x)\equiv\ell^{\LQ}(x^2,b_0)$, and the XMMSE--WMMSE gap $[1-\log(1+\gamma)/\gamma](x-x_0)^2$ of \cite{YangZhao26} is the SEFP--FP gap of \cite{Jiao26} written in the amplitude variable.
\end{proposition}
The proof, given in Appendix~\ref{app:xmmse}, evaluates the integrals of \cite{YangZhao26} for one stream and compares coefficients. Proposition~\ref{prop:xmmse} reveals that SEFP and XMMSE are the same surrogate derived from two directions---a tangent in the coordinate $1/r$ and an integral representation of the rate---so that one dominance theorem covers both. It also explains why the advantage of XMMSE over WMMSE grows with the SNR in \cite{YangZhao26}: the gap $(\gamma-\log(1+\gamma))(\sqrt\nu-1)^2$ grows linearly in $\gamma$, whereas the remaining gap to ICT grows only like $\log\gamma$ (Proposition~\ref{prop:curv}).

\begin{corollary}[Dominance]
\label{cor:dom}
For every anchor and every $(a,b)$, $S_{\LQ}=S_{\mathrm{WMMSE}}\le S_{\RQ}=S_{\mathrm{XMMSE}}\le S_{\ICT}\le\log(1+a/b)$, where $S_\bullet$ denotes the surrogate obtained from the respective single-ratio bound with its auxiliary variables fixed at the anchor. The same holds term by term for the WSR surrogates of Section~\ref{sec:prelim}.
\end{corollary}

Fig.~\ref{fig:surr} illustrates the own-signal functions \eqref{eq:philq}--\eqref{eq:phiict} at two SINRs. Observe that the LDT+QT bound is bounded above ($\phi_{\LQ}\le\log(1+\gamma)+1$, attained at $\nu=(1+\gamma)^2/\gamma^2$) and negative for small $\nu$; the RIT+QT bound removes the saturation at $\nu=0$ but retains a square-root shape; the ICT bound tracks the logarithm.

\begin{figure}[t]
\centering
\includegraphics[width=\columnwidth]{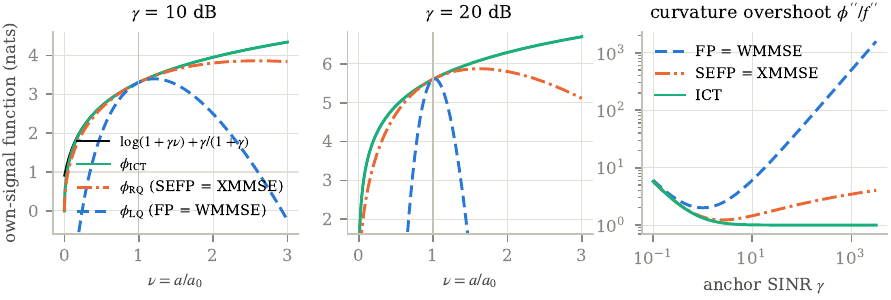}
\caption{Own-signal functions \eqref{eq:philq}--\eqref{eq:phiict} and the true rate (plus the constant $\gamma/(1+\gamma)$) versus the normalized own signal $\nu=a/a_0$, at $\gamma=10$\,dB (left) and $\gamma=20$\,dB (center). Right: curvature overshoot factors of Proposition~\ref{prop:curv} versus $\gamma$.}
\label{fig:surr}
\end{figure}

\subsection{Curvature, Certified Ascent and Asymptotic Rate}
\begin{proposition}[Own-signal curvature at the anchor]
\label{prop:curv}
With $f(\nu)=\log(1+\gamma\nu)$ (the rate at fixed interference), $f''(1)=-\gamma^2/(1+\gamma)^2$, and
\begin{equation}
\begin{aligned}
\frac{\phi_{\LQ}''(1)}{f''(1)}&=\frac{(1+\gamma)^2}{2\gamma},\qquad
\frac{\phi_{\RQ}''(1)}{f''(1)}=\frac{(1+\gamma)^2\log(1+\gamma)}{2\gamma^2},\\
\frac{\phi_{\ICT}''(1)}{f''(1)}&=1+\frac{1}{\gamma(\gamma+2)}.
\end{aligned}
\label{eq:curv}
\end{equation}
In the original variable, $\phi''_{\ICT}=-1/(a_0(a_0+2b_0))$, $\phi''_{\RQ}=-\log(1+\gamma)/(2a_0^2)$ and $\phi''_{\LQ}=-(1+\gamma)/(2a_0(a_0+b_0))$, against $f''=-1/(a_0+b_0)^2$.
\end{proposition}
The proof is given in Appendix~\ref{app:curv}. The three ratios behave as $\gamma/2$, $(\log\gamma)/2$ and $1$ for large $\gamma$: at a cell-centre SINR of 20--30\,dB, the FP/WMMSE surrogate is 50--500 times too curved in the own-signal direction, SEFP/XMMSE 2.4--3.5 times, and ICT is within $10^{-4}$ of the true curvature. Along the interference direction all three surrogates are linear, so that their curvature deficit there is identical; the differences between the methods are entirely in the own-signal direction.

\begin{proposition}[Certified ascent]
\label{prop:ascent}
Let $S_1\ge S_2$ be two MM surrogates of $f$ tight at $\bar{\mathbf x}$, and let $\mathbf x_i^+\in\argmax_{\mathbf x\in\mathcal X}S_i(\mathbf x;\bar{\mathbf x})$. Then
$f(\mathbf x_1^+)-f(\bar{\mathbf x})\ \ge\ \max_{\mathcal X}S_1-f(\bar{\mathbf x})\ \ge\ \max_{\mathcal X}S_2-f(\bar{\mathbf x})\ \ge0$. By Corollary~\ref{cor:dom}, the certified ascent of one ICT step is never smaller than that of one SEFP/XMMSE step, which in turn is never smaller than that of one FP/WMMSE step.
\end{proposition}
We emphasize that the proposition orders guarantees, not realized iterates: $f(\mathbf x_1^+)$ and $f(\mathbf x_2^+)$ can be in either order, because $f$ may exceed either surrogate by different amounts away from the anchor. The simulations in Sections~\ref{sec:siso}--\ref{sec:slqp} examine the realized behavior.

\begin{proposition}[Asymptotic linear rate]
\label{prop:rate}
Let $\mathbf x^\star$ be a stationary point of $f$ in the interior of $\mathcal X$ with $\mathbf H_f=\nabla^2f(\mathbf x^\star)\prec0$, and let $S_1\ge S_2$ be MM surrogates, twice differentiable and tight at every anchor, with Hessians $\mathbf H_1,\mathbf H_2$ in $\mathbf x$ at $(\mathbf x^\star;\mathbf x^\star)$. Then $\mathbf H_f\succeq\mathbf H_1\succeq\mathbf H_2$, the MM maps $T_i$ have Jacobians $\mathbf J_i=\mathbf I-\mathbf H_i^{-1}\mathbf H_f$ at $\mathbf x^\star$ with spectral radii
\begin{equation}
\varrho(\mathbf J_i)=1-\min_{\mathbf d\ne0}\frac{\mathbf d^T(-\mathbf H_f)\mathbf d}{\mathbf d^T(-\mathbf H_i)\mathbf d}\in[0,1),
\end{equation}
and $\varrho(\mathbf J_1)\le\varrho(\mathbf J_2)$. In SISO power control, $\mathbf H_i$ is diagonal with entries $w_j\phi_{j}''(a_j)g_{jj}^2$ (the interference terms being linear in $\mathbf p$), so that the ordering of Proposition~\ref{prop:curv} transfers entry-wise to the surrogate Hessians.
\end{proposition}
The proof is given in Appendix~\ref{app:rate}. Along the own-power direction of a link with SINR $\gamma$, the contraction factor $1-f''/\phi''$ is $\approx1-2/\gamma$ for FP/WMMSE, $\approx1-2/\log\gamma$ for SEFP/XMMSE and $1/(1+\gamma)^2$ for ICT. The first explains why closed-form FP needs thousands of iterations at 20--30\,dB, the second why SEFP/XMMSE need tens, and the third why the residual slowness of ICT, where present, comes from the shared interference directions rather than from the own-signal term.

\subsection{The Transformed Problem and Its Convergence}
Substituting \eqref{eq:ict} into \eqref{eq:wsr} gives the equivalent problem
\begin{equation}
\max_{\mathbf x\in\mathcal X,\,\boldsymbol\kappa\ge0}\ \sum_iw_i\big[\Om(\kappa_ia_i(\mathbf x))-\kappa_ib_i(\mathbf x)\big],
\label{eq:transformed}
\end{equation}
whose maximization over $\boldsymbol\kappa$ recovers \eqref{eq:wsr} exactly. Alternating between the closed-form update $\kappa_i=a_i/(b_i(a_i+b_i))$ and the maximization over $\mathbf x$ is an MM algorithm with the ICT surrogate. The degenerate anchor $a_i=0$ (zero price, flat surrogate) is excluded by $p_{\min}>0$, exactly as for FP and WMMSE.
\begin{proposition}[Convergence]
\label{prop:conv}
Let $\mathcal X$ be compact and let the $\mathbf x$ step maximize the ICT surrogate exactly (or, for the scheduling problem, maximize the separable surrogate over the schedule and the powers). Then the WSR is non-decreasing along the iterates and converges, and every limit point $\mathbf x^\star$ with $a_i(\mathbf x^\star)>0$ for all $i$ is a stationary point of \eqref{eq:wsr} (a KKT point when $\mathcal X$ is described by the box or power constraints of Section~\ref{sec:prelim}).
\end{proposition}
\begin{IEEEproof}
By Theorem~\ref{thm:ict}, the surrogate $S(\mathbf x;\bar{\mathbf x})=\sum_iw_i[\Om(\kappa_i^\star(\bar{\mathbf x})a_i(\mathbf x))-\kappa_i^\star(\bar{\mathbf x})b_i(\mathbf x)]$ satisfies $S\le f$ with equality at $\bar{\mathbf x}$, so that $f(\mathbf x^{t+1})\ge S(\mathbf x^{t+1};\mathbf x^t)\ge S(\mathbf x^t;\mathbf x^t)=f(\mathbf x^t)$; since $f$ is bounded on $\mathcal X$, $f(\mathbf x^t)$ converges. $S$ is continuously differentiable in $(\mathbf x,\bar{\mathbf x})$ wherever $a_i>0$, with $\nabla_{\mathbf x}S(\bar{\mathbf x};\bar{\mathbf x})=\nabla f(\bar{\mathbf x})$ by Proposition~\ref{prop:omega}(d). These are the hypotheses of the MM convergence theorem of \cite[Thm.~1]{Sun17} and \cite[Thm.~2]{Razaviyayn13}, which gives the stationarity of the limit points; for the scheduling problem, the argument of \cite[Prop.~1]{ShenYuII} applies verbatim with the ICT surrogate, since the schedule is updated by an exact maximization of the separable surrogate.
\end{IEEEproof}

\section{Power Control and Scheduling in SISO Networks}
\label{sec:siso}
\subsection{Power Control: Proposed Approach}
For fixed $\boldsymbol\kappa$, the surrogate of \eqref{eq:transformed} in the SISO power-control model is $\sum_iw_i\Om(\kappa_ig_{ii}p_i)-\sum_i\mu_ip_i+\mathrm{const}$, with $\mu_i=\sum_{j\ne i}w_j\kappa_jg_{ji}$ the total price that link $i$ pays for the interference it causes; it is concave and separable. Setting the derivative to zero, $w_i\kappa_ig_{ii}/\rhof(\kappa_ig_{ii}p_i)=\mu_i$, and using $\rhof^2/(1+\rhof)=u$, we obtain
\begin{equation}
\rho_i=\frac{w_i\kappa_ig_{ii}}{\mu_i},\qquad p_i=\Big[\frac{\rho_i^2}{(1+\rho_i)\kappa_ig_{ii}}\Big]_{p_{\min}}^{P_{\max}}.
\label{eq:pc}
\end{equation}
The resulting procedure is summarized in Algorithm~\ref{alg:pc}. Each iteration costs $O(K^2)$, exactly as closed-form FP \cite{ShenYuI} and SEFP \cite{Jiao26}, whose updates have the same structure with $\rho_i^2/(1+\rho_i)$ replaced by the corresponding QT expressions. The quantity $\rho_i$ has a direct interpretation: it is the SINR that link $i$ would attain if it responded optimally to the current prices, and $w_i\log(1+\rho_i)$ is the unclipped value of its subproblem. In power control, the ICT and FP updates differ in the value they produce rather than in its sign; the sign of the own-signal term matters in scheduling, where it decides whether a cell transmits at all (Section~\ref{sec:sched}).

\begin{algorithm}[t]
\caption{ICT power control (ICT-PC)}
\label{alg:pc}
\begin{algorithmic}[1]
\STATE Initialize $p_i\in[p_{\min},P_{\max}]$.
\REPEAT
\STATE $a_i\leftarrow g_{ii}p_i$, $b_i\leftarrow\sum_{j\ne i}g_{ij}p_j+\sigma^2$, $\kappa_i\leftarrow a_i/(b_i(a_i+b_i))$.
\STATE $\mu_i\leftarrow\sum_{j\ne i}w_j\kappa_jg_{ji}$, $\rho_i\leftarrow w_i\kappa_ig_{ii}/\mu_i$.
\STATE $p_i\leftarrow[\rho_i^2/((1+\rho_i)\kappa_ig_{ii})]_{p_{\min}}^{P_{\max}}$ for all $i$.
\UNTIL{the WSR has converged}
\end{algorithmic}
\end{algorithm}

\subsection{Power Control: Simulation Results}
\label{sec:num-pc}
\emph{Simulation methodology.} All simulations in this paper, except those of Section~\ref{sec:slqp}, use the propagation model of \cite{ShenYuI,ShenYuII,Jiao26}: 7 wrap-around hexagonal cells, path loss $128.1+37.6\log_{10}d$\,dB ($d$ in km), 8\,dB log-normal shadowing, Rayleigh fading, and noise $-169$\,dBm/Hz over 10\,MHz. The downlink power-control and beamforming settings use the 3GPP Case~1 layout \cite{TR36814}---inter-site distance (ISD) 500\,m and 46\,dBm per BS---and the D2D and uplink scheduling settings the 0.8\,km ISD and 23\,dBm of \cite{ShenYuII,Jiao26,Khan24}. The methods under comparison are always run on the same channel realizations (drops) from the same initial point under the same stopping rule, and all reported statistics are paired across drops. All methods use the power floor $p_{\min}=10^{-4}P_{\max}$ in power control, beamforming and percentile optimization, whereas the scheduling problems have $p_k\in[0,P_{\max}]$ as in \cite{ShenYuII,Jiao26}. Every identity of Sections~\ref{sec:ict}--\ref{sec:slqp} is verified numerically by the accompanying code: the ICT identity to \chkIdentity{}, the minorization for random $\kappa$ to \chkMinor{}, the dominance chain and $S_{\mathrm{WMMSE}}=S_{\LQ}$, $S_{\mathrm{XMMSE}}=S_{\RQ}$ to \chkWMMSE{} and \chkXMMSE{} on $2\cdot10^5$ random points, the curvature formulas \eqref{eq:curv} and the derivatives of Proposition~\ref{prop:dual} by finite differences; the beamformer of Proposition~\ref{prop:bf} is never exceeded by a brute-force search (gap \chkBfGap{}); and the spectral radius of Proposition~\ref{prop:rate} matches the empirical contraction factors on \chkRateN{} random 7-link instances to within \chkRateErr{} (mean factors \chkRateFP{}, \chkRateSEFP{} and \chkRateICT{} for FP, SEFP and ICT; the ICT factor is set by the interference directions, which all three surrogates flatten identically).

\emph{Setting.} We consider two settings with equal weights: (A) one downlink link per cell in the Case~1 layout ($K=7$, as in \cite[Fig.~3]{ShenYuI} but at 500\,m and 46\,dBm), and (B) $K=28$ device-to-device links dropped uniformly in the 0.8\,km layout with link lengths of 20--100\,m and 23\,dBm, an interference-limited regime. The initial powers are drawn uniformly at random in $[p_{\min},P_{\max}]$, as in \cite{Khan24}; \pcNdrops{} drops are simulated, and all methods run for \pcT{} iterations. The median SINR of the ICT solution is \pcCellSINR{}\,dB in (A) and \pcDenseSINR{}\,dB in (B), with 90th percentiles of \pcCellSINRp{} and \pcDenseSINRp{}\,dB. Table~\ref{tab:pc} and Fig.~\ref{fig:pc} report the results.

\begin{table*}[t]
\centering\small
\caption{SISO power control (\pcNdrops{} drops, \pcT{} iterations, random initial powers): median iterations to 99\% / 99.9\% of the method's own final value; mean final value relative to the best of the three on the same drop; drops within 0.1\% of the best; value at 50 iterations relative to the best final value; paired final-value comparison with ICT (mean, drops where ICT is better / worse by $>0.1\%$).}
\label{tab:pc}
\begin{tabular}{llcccrrc}
\toprule
& method & it.\ 99\% & it.\ 99.9\% & final vs best & within 0.1\% & at 50 it. & ICT vs method\\
\midrule
\multirow{3}{*}{(A) $K=7$}
& FP (=WMMSE) & \pcCellFPitNN & \pcCellFPitNNN & \pcCellFPrelBest\% & \pcCellFPwithin\% & \pcCellFPatFifty\% & \pcCellICTvsFPmean\% (\pcCellICTvsFPwins\% / \pcCellICTvsFPlosses\%)\\
& SEFP (=XMMSE) & \pcCellSEFPitNN & \pcCellSEFPitNNN & \pcCellSEFPrelBest\% & \pcCellSEFPwithin\% & \pcCellSEFPatFifty\% & \pcCellICTvsSEFPmean\% (\pcCellICTvsSEFPwins\% / \pcCellICTvsSEFPlosses\%)\\
& \textbf{ICT} & \textbf{\pcCellICTitNN} & \textbf{\pcCellICTitNNN} & \pcCellICTrelBest\% & \pcCellICTwithin\% & \pcCellICTatFifty\% & --\\
\midrule
\multirow{3}{*}{(B) $K=28$}
& FP (=WMMSE) & \pcDenseFPitNN & \pcDenseFPitNNN & \pcDenseFPrelBest\% & \pcDenseFPwithin\% & \pcDenseFPatFifty\% & \pcDenseICTvsFPmean\% (\pcDenseICTvsFPwins\% / \pcDenseICTvsFPlosses\%)\\
& SEFP (=XMMSE) & \pcDenseSEFPitNN & \pcDenseSEFPitNNN & \pcDenseSEFPrelBest\% & \pcDenseSEFPwithin\% & \pcDenseSEFPatFifty\% & \pcDenseICTvsSEFPmean\% (\pcDenseICTvsSEFPwins\% / \pcDenseICTvsSEFPlosses\%)\\
& \textbf{ICT} & \textbf{\pcDenseICTitNN} & \textbf{\pcDenseICTitNNN} & \pcDenseICTrelBest\% & \pcDenseICTwithin\% & \pcDenseICTatFifty\% & --\\
\bottomrule
\end{tabular}
\end{table*}

\begin{figure}[t]
\centering
\includegraphics[width=\columnwidth]{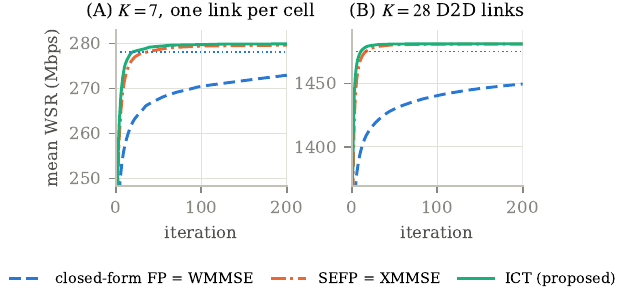}
\caption{SISO power control: mean WSR over the first 200 of the \pcT{} iterations in settings (A) and (B); the dotted line is the mean final value of FP after \pcT{} iterations.}
\label{fig:pc}
\end{figure}

\emph{Discussion.} We make three observations. First, the iteration counts follow Proposition~\ref{prop:rate}: closed-form FP needs a median of \pcCellFPitNNN{} (A) and \pcDenseFPitNNN{} (B) iterations to reach 99.9\% of its own final value, SEFP \pcCellSEFPitNNN{} and \pcDenseSEFPitNNN{}, and ICT \pcCellICTitNNN{} and \pcDenseICTitNNN{}; FP is still increasing at the end of the budget in \pcCellFPrising\% and \pcDenseFPrising\% of the drops, and its median projected-gradient residual $\|\mathbf p-\Pi[\mathbf p+P_{\max}\nabla f(\mathbf p)]\|_\infty/P_{\max}$ after \pcT{} iterations is \pcCellFPkkt{} in (A), against values \pcCellSEFPkkt{} for SEFP and ICT. The regime matters: sweeping the ISD of setting (A) at 46\,dBm (\isdN{} drops per point), FP needs a median of \isdFPitNearLo{}--\isdFPitNearHi{} iterations for ISDs up to 1\,km, \isdFPitTwo{} at 2\,km, \isdFPitFive{} at 5\,km and \isdFPitTen{} at 10\,km, where the network is noise-limited and all three methods reach the same point (within 0.1\% in at least \isdWithinTen\% of the drops); ICT needs \isdICTitMax{} to \isdICTitMin{}. Since scaling the ISD by $s$ is equivalent to scaling the transmit power by $s^{-3.76}$, the same holds for the power budget. Second, the methods do not always reach the same stationary point, and ICT reaches the better one on average: after \pcT{} iterations its final WSR exceeds that of FP by \pcCellICTvsFPmean\% in (A) and \pcDenseICTvsFPmean\% in (B), being higher by more than 0.1\% in \pcCellICTvsFPwins\% and \pcDenseICTvsFPwins\% of the drops and lower in only \pcCellICTvsFPlosses\% and \pcDenseICTvsFPlosses\% (part unfinished convergence, part worse points). ICT and SEFP, whose surrogates differ only by the curvature factor of Proposition~\ref{prop:curv}, reach the same points in nearly all drops (ICT ahead by \pcCellICTvsSEFPmean\% and \pcDenseICTvsSEFPmean\%), so that in equal-weight power control the advantage of ICT over SEFP is speed: half the iterations. Third, at a budget of 50 iterations the ordering is strict in both settings: ICT is within \pcCellICTatFifty\% / \pcDenseICTatFifty\% of the best final value, SEFP within \pcCellSEFPatFifty\% / \pcDenseSEFPatFifty\%, and FP within \pcCellFPatFifty\% / \pcDenseFPatFifty\%.

\subsection{Joint Uplink Scheduling and Power Control: Proposed Approach}
\label{sec:sched}
Shen and Yu \cite{ShenYuII} showed that the LDT+QT surrogate, after fixing the QT variables $y_i$, separates across cells into a sum over scheduled users of a reward minus the interference prices paid to the other cells, so that each BS can score its candidate users independently; Jiao \emph{et al.} \cite{Jiao26} retain this structure with RIT-induced coefficients. The ICT surrogate separates in the same way. Let $\lambda_i=w_{s_i}\kappa_i$ be the current price of cell $i$ (zero if the cell is off), and for a candidate $k\in\mathcal K_i$ let $\mu_k=\sum_{j\ne i}\lambda_jg_{j,k}$ be the price it would pay. Replacing the own-signal term of cell $i$ by $w_k\Om(\lambda_ig_{i,k}p_k/w_k)$ yields a valid, separable minorizer that is tight at the current schedule (Appendix~\ref{app:sched}); maximizing it over $p_k$ gives the candidate power and score
\begin{equation}
\begin{aligned}
\rho_k&=\frac{\lambda_ig_{i,k}}{\mu_k},\qquad
p_k=\min\Big\{P_{\max},\frac{w_k\rho_k^2}{(1+\rho_k)\lambda_ig_{i,k}}\Big\},\\
\Xi_i(k)&=w_k\log(1+\rho_k)
\end{aligned}
\label{eq:score}
\end{equation}
when the power is not clipped, and $\Xi_i(k)=w_k\Om(\lambda_ig_{i,k}P_{\max}/w_k)-\mu_kP_{\max}$ otherwise. Each BS then schedules $s_i=\argmax_{k\in\mathcal K_i}\Xi_i(k)$, or $\varnothing$ if all scores are non-positive.

\begin{proposition}[Structure of the three scheduling metrics]
\label{prop:sched}
With $\rho_k$ as in \eqref{eq:score}, the unclipped candidate scores of closed-form FP \cite[eqs.~(31)--(33)]{ShenYuII}, SEFP \cite[eqs.~(55)--(59)]{Jiao26} and ICT are
$\Xi_i(k)=w_k\,\ell(\rho_k;\gamma_i)$ with $\ell=\ell_{\mathrm{LDT}}$, $\ell_{\mathrm{RIT}}$ and $\log(1+\cdot)$, respectively, where $\gamma_i$ is the SINR of the user currently scheduled in cell $i$. Hence, the three methods charge the same prices $\lambda_i$ and differ only in $\ell$, and since $\ell_{\mathrm{LDT}}\le\ell_{\mathrm{RIT}}\le\log(1+\cdot)$, Proposition~\ref{prop:ascent} applies to each scheduling step.
\end{proposition}
The proposition is proved in Appendix~\ref{app:sched} by writing $y_i^2=\lambda_i$ for FP and $\alpha_i^2y_i^2=\lambda_i$ for SEFP. The ICT score is the exact best response of candidate $k$ to the current prices. In contrast, the LDT score $w_k[\log(1+\gamma_i)-\gamma_i+(1+\gamma_i)\rho_k/(1+\rho_k)]$ is bounded above by $w_k[\log(1+\gamma_i)+1]$ and is \emph{negative} whenever $\rho_k<(\gamma_i-\log(1+\gamma_i))/(1+\log(1+\gamma_i))$ (about $17$ at $\gamma_i=20$\,dB): a candidate with a large weight but a modest $\rho_k$ cannot be scheduled by FP, and if the condition holds for every candidate of a cell, FP switches the cell off for the rest of the slot. ICT and SEFP never switch a cell off, since their scores are positive for every $\rho_k>0$. We show in Section~\ref{sec:num-sched} that this is the mechanism behind the utility deficit of FP under proportional fairness. The ICT algorithm is Algorithm~1 of \cite{ShenYuII} with \eqref{eq:score} in place of its eqs.~(30)--(33); its complexity of $O(|\mathcal K||\mathcal B|)$ per iteration and its monotone convergence (Proposition~1 of \cite{ShenYuII}) are unchanged.

\subsection{Joint Uplink Scheduling and Power Control: Simulation Results}
\label{sec:num-sched}
We reproduce the setting of \cite[Sec.~IV-D]{ShenYuII} and \cite[Sec.~V]{Jiao26}: 84 single-antenna users dropped uniformly, associated with the strongest BS (every BS has at least one user), $P_{\max}=23$\,dBm ($-47$\,dBm/Hz), fading redrawn every slot, at most 50 iterations per slot, stopping when the WSR changes by less than $10^{-10}$. In every slot all three methods start from the same point---each cell schedules its highest-weight user at full power, the initialization that \cite{ShenYuII} uses for its MIMO algorithm---and the schedule and powers at the end of the slot determine the instantaneous rates. Under proportional fairness (PF), the weights are $w_k=1/\max(\bar R_k,0.01\mbps)$ with $\bar R_k$ the cumulative average rate of user $k$, and the utility is $\sum_k\log\bar R_k(T)$ over $T=\pfT$ slots (users never scheduled are floored at $10^{-12}$ as in \cite{Jiao26}). Absolute utilities depend on these conventions, which \cite{ShenYuII,Jiao26} do not fully specify; the paired differences are the meaningful quantity, and we note that our SEFP$-$FP difference reproduces the one reported in \cite{Jiao26}.

\begin{table}[t]
\centering\scriptsize\setlength{\tabcolsep}{4pt}
\caption{Uplink joint scheduling and power control. PF: log-utility, \pfNdrops{} drops $\times$ \pfT{} slots, with paired differences, 95\% confidence intervals and number of drops won. EW-SR: equal-weight sum rate over \ewN{} single-slot problems; NRP-WSR: normalized random-priority WSR over \nrpN{} problems (500 drops $\times$ 20 weight vectors). BCP: fraction of problems in which the method attains the best value.}
\label{tab:sched}
\begin{tabular}{lrrr}
\toprule
& FP \cite{ShenYuII} & SEFP \cite{Jiao26} & \textbf{ICT}\\
\midrule
PF utility (mean) & \pfFP & \pfSEFP & \textbf{\pfICT}\\
PF utility (std) & \pfStdFP & \pfStdSEFP & \pfStdICT\\
iterations per slot & \pfItersFP & \pfItersSEFP & \pfItersICT\\
slots at the 50-iteration cap & \pfCapFP\% & \pfCapSEFP\% & \pfCapICT\%\\
cells switched off (per slot) & \pfOffFP\% & \pfOffSEFP\% & \pfOffICT\%\\
5th-percentile user rate (Mbps) & \pfRateFP & \pfRateSEFP & \pfRateICT\\
\midrule
EW-SR (Mbps) & \ewFP & \ewSEFP & \textbf{\ewICT}\\
EW-SR BCP & \ewBestFP\% & \ewBestSEFP\% & \ewBestICT\%\\
NRP-WSR & \nrpFP & \nrpSEFP & \textbf{\nrpICT}\\
NRP-WSR BCP & \nrpBestFP\% & \nrpBestSEFP\% & \nrpBestICT\%\\
\midrule
\multicolumn{4}{l}{\emph{Paired differences (PF utility)}}\\
\multicolumn{3}{l}{SEFP $-$ FP} & \pfSEFPminusFP{} [\pfSEFPminusFPlo, \pfSEFPminusFPhi], \pfSEFPminusFPwins/\pfNdrops\\
\multicolumn{3}{l}{ICT $-$ FP} & \pfICTminusFP{} [\pfICTminusFPlo, \pfICTminusFPhi], \pfICTminusFPwins/\pfNdrops\\
\multicolumn{3}{l}{ICT $-$ SEFP} & \pfICTminusSEFP{} [\pfICTminusSEFPlo, \pfICTminusSEFPhi], \pfICTminusSEFPwins/\pfNdrops\\
\multicolumn{4}{l}{\emph{Paired differences (single-slot problems, mean \% and win/tie/loss \%)}}\\
\multicolumn{3}{l}{EW-SR: ICT vs FP} & +\ewICTvsFP\% \quad \ewICTvsFPwins/\ewICTvsFPties/\ewICTvsFPlosses\\
\multicolumn{3}{l}{EW-SR: ICT vs SEFP} & +\ewICTvsSEFP\% \quad \ewICTvsSEFPwins/\ewICTvsSEFPties/\ewICTvsSEFPlosses\\
\multicolumn{3}{l}{NRP-WSR: ICT vs FP} & +\nrpICTvsFP\% \quad \nrpICTvsFPwins/\nrpICTvsFPties/\nrpICTvsFPlosses\\
\multicolumn{3}{l}{NRP-WSR: ICT vs SEFP} & +\nrpICTvsSEFP\% \quad \nrpICTvsSEFPwins/\nrpICTvsSEFPties/\nrpICTvsSEFPlosses\\
\bottomrule
\end{tabular}
\end{table}

\emph{Proportional fairness.} Table~\ref{tab:sched} gives the main result. The SEFP$-$FP gain of \pfSEFPminusFP{} (\pfSEFPminusFPwins{} of \pfNdrops{} drops) replicates the $+2.996$ ($60/60$) reported in \cite[Table~II]{Jiao26}. ICT adds a further \pfICTminusSEFP{} over SEFP (\pfICTminusSEFPwins/\pfNdrops{} drops, 95\% CI [\pfICTminusSEFPlo, \pfICTminusSEFPhi]) and \pfICTminusFP{} over FP (\pfICTminusFPwins/\pfNdrops{}); that is, the gain of the tightest member of the family over SEFP is as large as the gain of SEFP over FP. ICT also uses the fewest iterations (\pfItersICT{} per slot against \pfItersSEFP{} and \pfItersFP{}, reaching the 50-iteration cap in \pfCapICT\% of the slots against \pfCapSEFP\% and \pfCapFP\%). The row ``cells switched off'' identifies the mechanism of Proposition~\ref{prop:sched}: FP switches off a cell in \pfOffFP\% of the slot decisions, whereas SEFP and ICT essentially never do. The three methods reach different PF operating points (median, mean, geometric-mean and 5th-percentile rates in Table~\ref{tab:sched}); the log-utility that all three maximize follows the geometric mean, \pfGeoICT{}\,Mbps under ICT against \pfGeoSEFP{} and \pfGeoFP{}.

\emph{Is the gain in the budget or in the points reached?} Repeating the PF experiment on \budN{} drops with a cap of 500 iterations, giving FP ten times the budget recovers only \budFPgain{} of its deficit, SEFP and ICT are already converged at 50 (changes of \budSEFPgain{} and \budICTgain{}), and ICT at 50 iterations beats FP at 500 by \budICTfiftyMinusFPfiveh{} (\budICTfiftyMinusFPfivehWins/\budN{}) and SEFP at 500 by \budICTfiftyMinusSEFPfiveh{}: the gain lies in the stationary points reached. On single-slot equal-weight problems (\convN{} drops, 1000 iterations), the median number of iterations to reach 99.9\% of the method's own gain is \convFPitNNN{} (FP), \convSEFPitNNN{} (SEFP) and \convICTitNNN{} (ICT), and the 1000-iteration limit of FP is \convFPlimit\% below the best.

\emph{Equal and random weights.} With equal weights (single-slot problems), the three methods reach nearly the same points: ICT attains the best value in \ewBestICT\% of \ewN{} problems and is \ewICTvsFP\% above FP and \ewICTvsSEFP\% above SEFP on average (Table~\ref{tab:sched}); here the difference between ICT and SEFP is mainly speed. With normalized random priorities (weights random per cell with per-cell sums equal to $|\mathcal K|/|\mathcal B|=12$ as in \cite{Jiao26}), where the weights differ across candidates, the gaps open: \nrpICTvsFP\% over FP and \nrpICTvsSEFP\% over SEFP, with ICT best in \nrpBestICT\% of \nrpN{} problems. This is consistent with Proposition~\ref{prop:sched}: the three metrics differ in how a candidate's weight interacts with the bounded own-signal term, which matters only when the weights are heterogeneous. With Fig.~\ref{fig:snr}, this locates where the exact score pays most: many candidates per cell with widely different weights and potential SINRs.

\emph{Sensitivity to the initial point.} On \initN{} drops, initializing each cell with the user of largest weighted link gain $w_kg_{i,k}$ instead of the highest-weight user costs FP \initFPdrop{} of utility (it then switches off \initOffWGFP\% of the cells per slot), against \initSEFPdrop{} for SEFP and \initICTdrop{} for ICT: the unbounded, positive ICT score always lets the schedule recover from a poor initial anchor within the slot; the bounded LDT score does not.

\emph{SNR sweep and load.} Following \cite[Sec.~V-C]{Jiao26}, Fig.~\ref{fig:snr} varies the nominal SNR at 0.4\,km from \snrLo{} to \snrHi{}\,dB by changing $P_{\max}$, with \snrN{} drops $\times$ \pfT{} slots per point, for the 12 users per cell of \cite{ShenYuII,Jiao26} and for 24. ICT attains the highest utility at every point and in every drop (\snrMinWins{}/\snrN{} paired wins over SEFP at every point and load). Expressed as the gain in geometric-mean user rate, $\exp(\Delta U/|\mathcal K|)-1$, its advantage over SEFP is largest in noise-limited, heavily loaded cells: \snrGmSEFPlo{} at \snrLo{}\,dB against \snrGmSEFPhi{} at \snrHi{}\,dB with 12 users per cell, and \snrGmSEFPloB{} against \snrGmSEFPhiB{} with 24; over FP it is \snrGmFPlo{} to \snrGmFPhi{} and \snrGmFPloB{} to \snrGmFPhiB{}. The mechanism is that of Proposition~\ref{prop:sched}: the RIT score is anchored at the SINR $\gamma_i$ of the user currently scheduled and saturates at $c(\gamma_i)/b(\gamma_i)$---about 2 nats when that user is at low SINR---so that it under-rewards candidates whose best-response SINR is far from $\gamma_i$, which are most numerous in noise-limited, heavily loaded cells; the exact ICT score has no cap. Proposition~\ref{prop:curv} concerns speed rather than the points reached, and its SNR dependence shows in the iteration counts: \snrItFPhi{}, \snrItSEFPhi{} and \snrItICThi{} per slot (FP, SEFP, ICT) at \snrHi{}\,dB against \snrItFPlo{}, \snrItSEFPlo{} and \snrItICTlo{} at \snrLo{}\,dB.

\begin{figure}[t]
\centering
\includegraphics[width=\columnwidth]{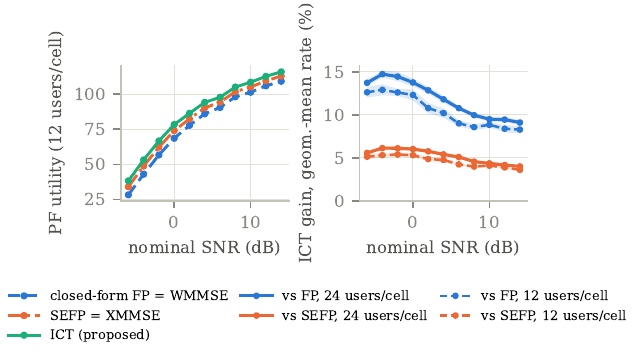}
\caption{PF utility versus nominal SNR with 12 users per cell (left) and paired gain of ICT over FP and SEFP in geometric-mean user rate, with 95\% confidence bands, for 12 and 24 users per cell (right); \snrN{} drops $\times$ \pfT{} slots per point.}
\label{fig:snr}
\end{figure}

\section{Multi-Cell MISO Beamforming}
\label{sec:miso}
\subsection{Proposed Approach}
For MISO beamforming, $a_i=|\mathbf h_{ii}^H\mathbf v_i|^2$ is quadratic and $b_i$ is a sum of quadratics in the other beamformers, so that the ICT surrogate at the anchor $\{\bar{\mathbf v}_j\}$ is
\begin{equation}
\sum_jw_j\Om\big(\kappa_j|\mathbf h_{jj}^H\mathbf v_j|^2\big)-\sum_j\mathbf v_j^H\mathbf C_j\mathbf v_j+\mathrm{const},
\label{eq:bfsurr}
\end{equation}
with $\mathbf C_j=\sum_{i\ne j}w_i\kappa_i\mathbf h_{ij}\mathbf h_{ij}^H$, which separates across BSs. The per-BS problem $\max_{\|\mathbf v\|^2\le P}w\Om(\kappa|\mathbf h^H\mathbf v|^2)-\mathbf v^H\mathbf C\mathbf v$ is not concave in $\mathbf v$; nevertheless, it admits a closed-form global maximizer.
\begin{proposition}[Closed-form ICT beamformer]
\label{prop:bf}
Let $\mathbf C=\mathbf Q\diag(\mathbf c)\mathbf Q^H\succeq0$, $\tilde{\mathbf h}=\mathbf Q^H\mathbf h$ and, for $\eta\ge0$, $q(\eta)=\sum_m|\tilde h_m|^2/(c_m+\eta)$, $\rho(\eta)=w\kappa q(\eta)$, $s(\eta)=\rho^2/(\kappa(1+\rho))$ and
\begin{equation}
\mathbf v(\eta)=\frac{\sqrt{s(\eta)}}{q(\eta)}\,(\mathbf C+\eta\mathbf I)^{-1}\mathbf h.
\label{eq:vbf}
\end{equation}
Then $\|\mathbf v(\eta)\|^2$ is non-increasing in $\eta$, and the global maximizer of the per-BS problem is $\mathbf v(0)$ if $\|\mathbf v(0)\|^2\le P$ and otherwise $\mathbf v(\eta^\star)$, with $\eta^\star\in(0,w/P]$ the root of $\|\mathbf v(\eta)\|^2=P$. The unclipped value is $w\log(1+\rho(0))$.
\end{proposition}
\begin{IEEEproof}
For fixed $\eta$, write $s=|\mathbf h^H\mathbf v|^2$. Among all $\mathbf v$ with this $s$, $\mathbf v^H(\mathbf C+\eta\mathbf I)\mathbf v$ is minimized by $\mathbf v\propto(\mathbf C+\eta\mathbf I)^{-1}\mathbf h$ with minimum $s/q(\eta)$ (Cauchy--Schwarz), so that the Lagrangian $L(\mathbf v,\eta)=w\Om(\kappa|\mathbf h^H\mathbf v|^2)-\mathbf v^H(\mathbf C+\eta\mathbf I)\mathbf v$ is maximized by solving the concave scalar problem $\max_{s\ge0}w\Om(\kappa s)-s/q$, whose stationarity condition $w\kappa/\rhof(\kappa s)=1/q$ gives $\rhof=\rho(\eta)$ and $s=s(\eta)$; this yields \eqref{eq:vbf}. The function $\eta\mapsto\max_{\mathbf v}L(\mathbf v,\eta)+\eta P$ is convex with derivative $P-\|\mathbf v(\eta)\|^2$, so that $\|\mathbf v(\eta)\|^2$ is non-increasing; the bound $\|\mathbf v(\eta)\|^2\le s/(q\eta)\le w/\eta$ gives the bracket. If $\eta^\star\ge0$ satisfies complementary slackness, then for every feasible $\mathbf v$, $f(\mathbf v)\le L(\mathbf v,\eta^\star)\le L(\mathbf v(\eta^\star),\eta^\star)=f(\mathbf v(\eta^\star))$, which proves global optimality despite the non-concavity.
\end{IEEEproof}
The ICT-BF iteration (Algorithm~\ref{alg:bf}) therefore costs one $M\times M$ eigendecomposition and one scalar bisection per BS, which is the cost of the closed-form FP beamformer $\mathbf v_j=\sqrt{w_j(1+\gamma_j)}\,y_j(\sum_i|y_i|^2\mathbf h_{ij}\mathbf h_{ij}^H+\eta\mathbf I)^{-1}\mathbf h_{jj}$ of \cite{ShenYuI}, of WMMSE, and of XMMSE. Specifically, the per-BS updates of the four methods have the common form $\mathbf v_j=\beta_j(\mathbf C_j+\theta_j\mathbf h_{jj}\mathbf h_{jj}^H+\eta\mathbf I)^{-1}\mathbf h_{jj}$ with $\mathbf C_j=\sum_{i\ne j}w_i\kappa_i^\star\mathbf h_{ij}\mathbf h_{ij}^H$, a scalar $\beta_j$, and an own-channel coefficient $\theta_j$ equal to the curvature of the own-signal term: $\theta_j=w_j\kappa_j^\star$ for FP/WMMSE (the QT flattens the own signal with the same weight as the interference), $\theta_j=w_jb(\gamma_j)/(\gamma_j(1+\gamma_j)b_j)$ for XMMSE/RIT+QT, which is smaller by the factor $b(\gamma)/\gamma^2\approx\log\gamma/\gamma$, and $\theta_j=0$ for ICT, whose $\beta_j=\sqrt{s(\eta)}/q(\eta)$ depends on $\eta$ through $\Om$ instead. This is the algorithmic face of Proposition~\ref{prop:curv}.

\begin{algorithm}[t]
\caption{ICT beamforming (ICT-BF)}
\label{alg:bf}
\begin{algorithmic}[1]
\STATE Initialize $\mathbf v_j$ with $\|\mathbf v_j\|^2\le P$.
\REPEAT
\STATE $a_i\leftarrow|\mathbf h_{ii}^H\mathbf v_i|^2$, $b_i\leftarrow\sum_{j\ne i}|\mathbf h_{ij}^H\mathbf v_j|^2+\sigma^2$, $\kappa_i\leftarrow a_i/(b_i(a_i+b_i))$.
\STATE For each BS $j$: $\mathbf C_j\leftarrow\sum_{i\ne j}w_i\kappa_i\mathbf h_{ij}\mathbf h_{ij}^H$; $\mathbf v_j\leftarrow\mathbf v(\eta^\star)$ from Proposition~\ref{prop:bf} with $(\mathbf C,\mathbf h,w,\kappa)=(\mathbf C_j,\mathbf h_{jj},w_j,\kappa_j)$; rescale to the power floor if $\|\mathbf v_j\|^2<p_{\min}$.
\UNTIL{the WSR has converged}
\end{algorithmic}
\end{algorithm}

The beamforming surrogates of closed-form FP and WMMSE are quadratic in the $\mathbf v_j$: with $z_j=\mathbf h_{jj}^H\mathbf v_j$, the price $|y_i|^2=w_i\kappa_i^\star$ and the own-signal term $2\Real(y_j^*\sqrt{w_j(1+\gamma_j)}\,z_j)-|y_j|^2|z_j|^2$, they are $w_j\ell^{\LQ}$ evaluated at $a=|z_j|^2$, except that the linear term $2\gamma\sqrt\nu$ of \eqref{eq:philq} is replaced by $2\gamma\,\Real(z_{j,0}^*z_j)/|z_{j,0}|^2\le2\gamma\sqrt\nu$. By Proposition~\ref{prop:xmmse}, the single-stream XMMSE surrogate is \eqref{eq:xmmse-vec}, with the same price and the same modification of $\phi_{\RQ}$. Both own-signal terms are therefore bounded above by $w_j\phi_{\LQ}(\nu_j)$ and $w_j\phi_{\RQ}(\nu_j)$ with $\nu_j=|z_j|^2/|z_{j,0}|^2$, which Theorem~\ref{thm:tightest} bounds by $w_j\phi_{\ICT}(\nu_j)$, the ICT-BF own-signal term. Hence, we have the following corollary.
\begin{corollary}
\label{cor:bf}
For every anchor and every set of beamformers, the WSR surrogates satisfy $S_{\mathrm{FP}}=S_{\mathrm{WMMSE}}\le S_{\mathrm{XMMSE}}=S_{\RQ}\le S_{\ICT}\le\mathrm{WSR}$, and the ICT-BF surrogate is the tightest surrogate of the form $\sum_j[\phi_j(\mathbf h_{jj}^H\mathbf v_j)-\lambda_jb_j]$.
\end{corollary}
We have verified numerically that closed-form FP and WMMSE produce identical beamformers (maximum relative deviation \bfMfourEqFPW{} over 300 iterations, consistent with \cite{Zhang23}) and, which is new, that XMMSE and RIT+QT beamforming produce identical beamformers (\bfMfourEqXR{}), so that the MISO comparison below has two baselines.

\begin{remark}[Single-stream MIMO users]
\label{rem:mimo}
If user $i$ has $N$ receive antennas and one stream, its rate is $\log(1+\boldsymbol\alpha^H\mathbf B^{-1}\boldsymbol\alpha)$ with $\boldsymbol\alpha=\mathbf H_{ii}\mathbf v_i\in\C^N$ and $\mathbf B\succ0$ the interference-plus-noise covariance. Among surrogates $\phi(\boldsymbol\alpha)-\tr(\mathbf K\mathbf B)$ affine in $\mathbf B$, tightness at $(\boldsymbol\alpha_0,\mathbf B_0)$ forces the rank-one price $\mathbf K_0=\mathbf B_0^{-1}\boldsymbol\alpha_0\boldsymbol\alpha_0^H\mathbf B_0^{-1}/(1+\gamma_0)$, and by Cauchy--Schwarz, $\inf_{\mathbf B\succ0}[\log(1+\boldsymbol\alpha^H\mathbf B^{-1}\boldsymbol\alpha)+\tr(\mathbf K_0\mathbf B)]=\Om(|\mathbf u_0^H\boldsymbol\alpha|^2/(1+\gamma_0))$ with $\mathbf u_0=\mathbf B_0^{-1}\boldsymbol\alpha_0$ (Appendix~\ref{app:mimo}). The matrix case thus collapses to the scalar $\Om$ with no Riccati equation, the ICT beamformer is \eqref{eq:vbf} with $\mathbf h$ replaced by $\mathbf H_{ii}^H\mathbf u_0$ up to scaling, and the single-stream MIMO version of XMMSE, whose price is $w_i\mathbf K_0$ by inspection of its coefficients, is dominated as well. Multi-stream users with $\log\det$ rates are not covered: there the infimum over $\mathbf B$ does not collapse to a scalar function.
\end{remark}

\subsection{Simulation Results}
\label{sec:num-bf}
We use the setting of \cite[Sec.~VI]{ShenYuI} in the Case~1 layout: one user per cell, $M=4$ (\bfMfourN{} drops) or $M=8$ (\bfMeightN{} drops) antennas per BS, 46\,dBm per BS, equal weights, and all methods initialized with maximum-ratio transmission at full power. The SINRs are high in this setting (median SINR of the ICT solution \bfMfourSINR{}\,dB and \bfMeightSINR{}\,dB, 90th percentiles \bfMfourSINRp{} and \bfMeightSINRp{}\,dB), which is the regime in which Proposition~\ref{prop:curv} predicts the largest differences. All methods are run for \bfMfourT{} ($M=4$) and \bfMeightT{} ($M=8$) iterations; Fig.~\ref{fig:bf} and Table~\ref{tab:bf} report the results. A third configuration with $M=4$ and weights drawn uniformly in $[0.1,1]$ (\bfMfourWN{} drops) gives the same picture (median iterations to 99.9\%: \bfMfourWFPitNNN{} / \bfMfourWXMitNNN{} / \bfMfourWICTitNNN{}; final value of ICT \bfMfourWICTvsFPmean\% above FP/WMMSE and \bfMfourWICTvsXMmean\% above XMMSE) and is omitted from the table.

\begin{figure}[t]
\centering
\includegraphics[width=\columnwidth]{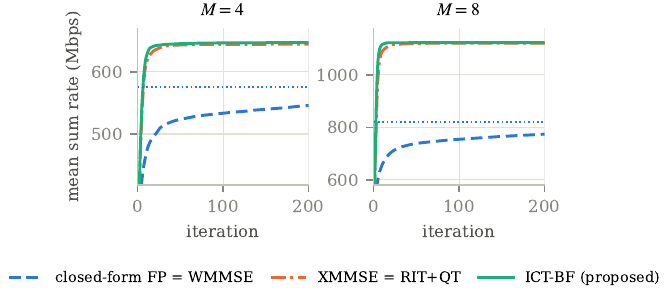}
\caption{MISO beamforming: mean sum rate over the first 200 iterations for $M=4$ (left) and $M=8$ (right); the dotted line is the mean final value of FP/WMMSE after \bfMfourT{} / \bfMeightT{} iterations.}
\label{fig:bf}
\end{figure}

\begin{table*}[t]
\centering\small
\caption{MISO beamforming ($M=4$: \bfMfourN{} drops, \bfMfourT{} it.; $M=8$: \bfMeightN{} drops, \bfMeightT{} it.): median iterations to 99\% / 99.9\% of the method's own final value; mean final value relative to the best of the three on the same drop; drops within 0.1\% of the best; value after 50 and 300 iterations relative to the best final value; mean time per iteration per drop (NumPy, one core).}
\label{tab:bf}
\begin{tabular}{llcccrrrc}
\toprule
& method & 99\% & 99.9\% & final vs best & within 0.1\% & at 50 & at 300 & ms/it.\\
\midrule
\multirow{3}{*}{$M=4$}
& FP=WMMSE & \bfMfourFPitNN & \bfMfourFPitNNN & \bfMfourFPrelBest\% & \bfMfourFPwithin\% & \bfMfourFPatFifty\% & \bfMfourFPatThreeH\% & \bfMfourFPtime\\
& XMMSE=RIT+QT & \bfMfourXMitNN & \bfMfourXMitNNN & \bfMfourXMrelBest\% & \bfMfourXMwithin\% & \bfMfourXMatFifty\% & \bfMfourXMatThreeH\% & \bfMfourXMtime\\
& \textbf{ICT-BF} & \textbf{\bfMfourICTitNN} & \textbf{\bfMfourICTitNNN} & \bfMfourICTrelBest\% & \bfMfourICTwithin\% & \bfMfourICTatFifty\% & \bfMfourICTatThreeH\% & \bfMfourICTtime\\
\midrule
\multirow{3}{*}{$M=8$}
& FP=WMMSE & \bfMeightFPitNN & \bfMeightFPitNNN & \bfMeightFPrelBest\% & \bfMeightFPwithin\% & \bfMeightFPatFifty\% & \bfMeightFPatThreeH\% & \bfMeightFPtime\\
& XMMSE=RIT+QT & \bfMeightXMitNN & \bfMeightXMitNNN & \bfMeightXMrelBest\% & \bfMeightXMwithin\% & \bfMeightXMatFifty\% & \bfMeightXMatThreeH\% & \bfMeightXMtime\\
& \textbf{ICT-BF} & \textbf{\bfMeightICTitNN} & \textbf{\bfMeightICTitNNN} & \bfMeightICTrelBest\% & \bfMeightICTwithin\% & \bfMeightICTatFifty\% & \bfMeightICTatThreeH\% & \bfMeightICTtime\\
\bottomrule
\end{tabular}
\end{table*}

\emph{Convergence.} Fig.~\ref{fig:bf} illustrates the convergence of the mean sum rate. The median number of iterations to 99.9\% of the method's own final value is \bfMfourFPitNNN{} / \bfMfourXMitNNN{} / \bfMfourICTitNNN{} for $M=4$ and \bfMeightFPitNNN{} / \bfMeightXMitNNN{} / \bfMeightICTitNNN{} for $M=8$ (FP/WMMSE, XMMSE, ICT). FP/WMMSE is still increasing at the end of its budget in \bfMfourFPrising\% ($M=4$) and \bfMeightFPrising\% ($M=8$) of the drops, with a relative gain per iteration of about $10^{-5}$; at the budget of 300 iterations that one would typically allow, its sum rate is \bfMfourFPatThreeH\% ($M=4$) and \bfMeightFPatThreeH\% ($M=8$) below the best final value. XMMSE removes most of this slowness, as its curvature factor $(\log\gamma)/2$ predicts, and ICT removes the rest: at 50 iterations ICT is within \bfMfourICTatFifty\% / \bfMeightICTatFifty\% of the best final value, against \bfMfourXMatFifty\% / \bfMeightXMatFifty\% for XMMSE.

\emph{Points reached.} After the full budget, the final value of ICT exceeds that of FP/WMMSE by more than 0.1\% in \bfMfourICTvsFPwins\% ($M=4$) and \bfMeightICTvsFPwins\% ($M=8$) of the drops and is lower in \bfMfourICTvsFPlosses\% and \bfMeightICTvsFPlosses\%; its mean final value is \bfMfourICTvsFPmean\% and \bfMeightICTvsFPmean\% higher. Against XMMSE the differences are small but one-sided: \bfMfourICTvsXMmean\% ($M=4$; better by more than $0.1\%$ in \bfMfourICTvsXMwins\% of the drops, worse in \bfMfourICTvsXMlosses\%) and \bfMeightICTvsXMmean\% ($M=8$; \bfMeightICTvsXMwins\% / \bfMeightICTvsXMlosses\%). ICT at 50 iterations already exceeds FP/WMMSE at the full budget by \bfMfourICTatFiftyVsFPfinal\% ($M=4$) and \bfMeightICTatFiftyVsFPfinal\% ($M=8$), and the per-iteration times of the three methods are within a factor of 1.5 (Table~\ref{tab:bf}), so that iteration counts translate into wall-clock time.

\section{Solver-Free Percentile Optimization}
\label{sec:slqp}
\subsection{Problem Formulation}
The sum-least-$q$th-percentile (SLqP) rate problem of Khan and Adve \cite{Khan24} maximizes the sum of the $K_q$ smallest user rates,
\begin{equation}
\max_{\mathbf p\in[p_{\min},P_{\max}]^K}\ f_{K_q}(\mathbf r(\mathbf p)),\qquad f_{K_q}(\mathbf r)=\min_{\mathbf x\in\mathcal X}\mathbf x^T\mathbf r,
\label{eq:slqp}
\end{equation}
with $\mathcal X=\{\mathbf x\in[0,1]^K:\mathbf 1^T\mathbf x=K_q\}$ the capped simplex (its vertices are the $0/1$ selections of $K_q$ users, so that the minimum is the sum of the $K_q$ smallest rates). It interpolates between max-min ($K_q=1$) and sum-rate ($K_q=K$) and is strongly NP-hard for $K_q>1$ \cite{Khan24}. The QFT and LFT of \cite{Khan24} are MM methods whose subproblems are convex but admit no closed form and are solved by a conic solver at every iteration; the authors state that closed-form algorithms for SLqP are unknown \cite{Khan24}.

\subsection{Proposed Approach}
The ICT surrogate supplies such an algorithm. At the anchor $\bar{\mathbf p}$, replace $r_k$ by $S_k(\mathbf p)=\Om(\kappa_kg_{kk}p_k)-\kappa_kb_k(\mathbf p)$. The MM subproblem $\max_{\mathbf p}\min_{\mathbf x\in\mathcal X}\mathbf x^T\mathbf S(\mathbf p)$ is concave in $\mathbf p$ and linear in $\mathbf x$ on compact convex sets, so that by Sion's minimax theorem it equals $\min_{\mathbf x\in\mathcal X}d(\mathbf x)$ with the dual function
\begin{equation}
d(\mathbf x)=\max_{\mathbf p}\sum_k\big[x_k\Om(\kappa_kg_{kk}p_k)-p_k\mu_k(\mathbf x)\big]-\sigma^2\boldsymbol\kappa^T\!\mathbf x
\label{eq:dual}
\end{equation}
with $\mu_k(\mathbf x)=\sum_{j\ne k}x_j\kappa_jg_{jk}$ and the maximization over $[p_{\min},P_{\max}]^K$. Because the surrogate is affine in the interference and separable in the own powers, the maximization in \eqref{eq:dual} is \eqref{eq:pc} with the weights $w_k$ replaced by $x_k$: $\mathbf p(\mathbf x)$ is closed form, and so are the derivatives of $d$.
\begin{proposition}[Dual of the ICT-SLqP subproblem]
\label{prop:dual}
$d$ is convex on $\mathcal X$, with $\nabla d(\mathbf x)=\mathbf S(\mathbf p(\mathbf x))$ and $\nabla^2d(\mathbf x)=\mathbf J_S\,\partial\mathbf p/\partial\mathbf x$, where $(\mathbf J_S)_{kj}=\delta_{kj}\phi_k'(p_k)-(1-\delta_{kj})\kappa_kg_{kj}$, $\phi_k(p)=\Om(\kappa_kg_{kk}p)$, and, for the free coordinates, $\partial p_k/\partial x_j=[(1-\delta_{kj})\kappa_jg_{jk}-\delta_{kj}\phi_k'(p_k)]/(x_k\phi_k''(p_k))$. For every $\mathbf x\in\mathcal X$, $f_{K_q}(\mathbf S(\mathbf p(\mathbf x)))\le\max_{\mathbf p}\min_{\mathbf x'}\mathbf x'^T\mathbf S(\mathbf p)\le d(\mathbf x)$, so that the difference $d(\mathbf x)-f_{K_q}(\mathbf S(\mathbf p(\mathbf x)))$ certifies the suboptimality of $\mathbf p(\mathbf x)$ for the subproblem.
\end{proposition}
\begin{IEEEproof}
$d$ is a maximum of functions affine in $\mathbf x$, hence convex; the gradient follows from the envelope theorem (the maximizer is unique, and clipped coordinates have zero derivative); the Hessian follows by differentiating the stationarity condition $x_k\phi_k'(p_k)=\mu_k(\mathbf x)$ of the free coordinates. The bounds are weak duality and the definition of $f_{K_q}$.
\end{IEEEproof}
Consequently, the subproblem is a $K$-dimensional smooth convex problem over a capped simplex, which Algorithm~\ref{alg:slqp} solves by a barrier Newton method: each step is one $K\times K$ linear solve, the gradient and Hessian are closed form, and the certified gap of Proposition~\ref{prop:dual} serves as the stopping rule. No conic solver and no modeling layer are involved. We emphasize that the construction uses exactly the two properties of Section~\ref{sec:ict}: affineness in the interference (so that the inner maximization separates and the dual has closed-form derivatives) and tightness (so that the outer MM loop needs as few iterations as the solver-based transforms). The first property is shared by the LDT+QT and RIT+QT surrogates, which can be used in the same way; Section~\ref{sec:num-slqp} compares them. Convergence to directional stationary points of \eqref{eq:slqp} follows as for QFT in \cite[Thm.~4]{Khan24}, since the ICT surrogate is a tangent minorant with matching directional derivatives. The same dual applies to SLqP beamforming \cite{KhanAdveII}: for fixed $\mathbf x$ the maximization separates across BSs, each beamformer is given by Proposition~\ref{prop:bf} with $w_j=x_j$, and $\nabla d=\mathbf S(\mathbf v(\mathbf x))$ as before.

\begin{algorithm}[t]
\caption{Solver-free ICT-SLqP}
\label{alg:slqp}
\begin{algorithmic}[1]
\STATE Initialize $\mathbf p$; $\mathbf x\leftarrow(K_q/K)\mathbf 1$.
\REPEAT
\STATE $\kappa_k\leftarrow a_k/(b_k(a_k+b_k))$ at the current $\mathbf p$ (ICT surrogate).
\STATE Minimize $d$ over $\mathcal X$ by barrier Newton steps (warm start $\mathbf x$): at each step evaluate $\mathbf p(\mathbf x)$ by \eqref{eq:pc} with $w_k=x_k$, $\nabla d$ and $\nabla^2d$ by Proposition~\ref{prop:dual}; stop when $d(\mathbf x)-f_{K_q}(\mathbf S(\mathbf p(\mathbf x)))\le10^{-9}(1+|d|)$.
\STATE $\mathbf p\leftarrow\mathbf p(\mathbf x)$.
\UNTIL{$f_{K_q}(\mathbf r(\mathbf p))$ has converged}
\end{algorithmic}
\end{algorithm}

\subsection{Simulation Results}
\label{sec:num-slqp}
Table~\ref{tab:slqp} compares Algorithm~\ref{alg:slqp} with the QFT and LFT algorithms of \cite{Khan24}, whose convex subproblems we solve by the interior-point conic solver Clarabel through CVXPY (with the variables rescaled so that the noise power is one; without rescaling, the exponential-cone problems fail), and with the LDT+QT and RIT+QT surrogates solved through the same dual machinery (ablation). Two settings follow \cite[Sec.~IV-C]{Khan24}: 7 wrap-around cells with 2\,km inter-site distance, path loss $(1+d/d_0)^{-\zeta}$ with $d_0=0.392$\,m and $\zeta=3.76$, Rayleigh fading, $-143$\,dBm/Hz noise over 20\,MHz, 43\,dBm, uplink, with $K=56$ users and $K_q=14$ (25th percentile) or $K=70$ and $K_q=7$ (10th percentile), and random initial powers; this regime is noise-limited (median SINR of the solution \slqpKAfiftysixSINR{}\,dB). The third setting is interference-limited: 4 uplink users per cell with the propagation model of Section~\ref{sec:num-pc}, 23\,dBm, $K=28$, $K_q=7$ (median SINR \slqpSYSINR{}\,dB). Every MM loop stops when the relative change of the SLqP objective is below $10^{-6}$ on two consecutive iterations or after 200 iterations, with \slqpN{} drops in the first two settings and \slqpSYN{} in the third.

\begin{table*}[t]
\centering\scriptsize
\caption{SLqP rate maximization. Median outer iterations to 99\% / 99.9\% of the method's own final value; mean final value relative to the best of the five methods on the same drop and fraction of drops within 0.1\% of it; mean time per outer iteration and total time to the stopping rule (NumPy / CVXPY+Clarabel, one core).}
\label{tab:slqp}
\begin{tabular}{llcccccrr}
\toprule
& method & inner solver & it.\ 99\% & it.\ 99.9\% & final vs best & within 0.1\% & ms/it. & total s\\
\midrule
\slqpTableRows
\bottomrule
\end{tabular}
\vspace{1.5em}
\end{table*}

We draw three conclusions. First, in the settings of \cite{Khan24}, the solver-free algorithm reaches the same values as the solver-based QFT (within 0.1\% of the best of the five methods in \slqpKAfiftysixICTwithin\% and \slqpKAseventyICTwithin\% of the drops, against \slqpKAfiftysixQFTwithin\% and \slqpKAseventyQFTwithin\% for QFT) in fewer outer iterations (median \slqpKAfiftysixICTitNNN{} and \slqpKAseventyICTitNNN{} to 99.9\%, against \slqpKAfiftysixQFTitNNN{} and \slqpKAseventyQFTitNNN{}); LFT does not converge within 200 iterations in this noise-limited regime and ends \slqpKAfiftysixLFTrelBest\% and \slqpKAseventyLFTrelBest\% below the best. In the interference-limited setting, the SLqP landscape has many stationary points and the methods land on different ones: ICT is better than QFT by more than 0.1\% in \slqpSYICTvsQFTwins\% of \slqpSYN{} drops and worse in \slqpSYICTvsQFTlosses\% (mean difference \slqpSYICTvsQFTmean\%), and the same holds between ICT and the two other dual methods (\slqpSYICTvsSEFPwins\% / \slqpSYICTvsSEFPlosses\% against RIT+QT). Cross-restarts confirm that these are genuinely different stationary points (QFT started from an ICT solution does not move, and conversely). It is worth emphasizing that tightness within the interference-affine family orders surrogates and certified ascent, not the stationary points of a non-smooth non-convex objective, and that the QFT surrogate lies outside the family. Second, the dual machinery works for any interference-affine surrogate, and the subproblem values are ordered as Proposition~\ref{prop:ascent} predicts: at a common anchor, the maximum of the ICT subproblem exceeds that of the RIT+QT subproblem by \slqpAscentSEFP\% and that of the LDT+QT subproblem by \slqpAscentFP\% on average, with no violation of the ordering in \slqpAscentN{} anchors, and the LDT+QT surrogate needs the most outer iterations (median \slqpKAfiftysixFPitNNN{} and \slqpSYFPitNNN{}, against \slqpKAfiftysixICTitNNN{} and \slqpSYICTitNNN{} for ICT). Third, regarding the computational cost, at $K=56$--$70$ the dual Newton subproblem (\slqpKAfiftysixICTNewton{} Newton steps on average, each a $K\times K$ solve) is \slqpKAfiftysixICTvsQFTtimeRatio{}--\slqpKAseventyICTvsQFTtimeRatio{} times faster than one CVXPY/Clarabel solve, and Table~\ref{tab:scaling} shows how this advantage grows with the network size: at $K=140$ and $K=280$, the certified dual solve is \scalingRatioOneForty{}$\times$ and \scalingRatioTwoEighty{}$\times$ faster than the interior-point solve of Clarabel, which at $K=280$ reports an inaccurate solution on every instance, whereas the first-order conic solver SCS is \scalingScsRatioFiftySix{}$\times$ slower than the dual solve at $K=56$ and \scalingScsRatioOneForty{}$\times$ at $K=140$. For reference, \cite[Table~I]{Khan24} reports 13\,s per instance for QFT with CVX at $K=21$ on a 12-core workstation, whereas Algorithm~\ref{alg:slqp} converges in \slqpKAfiftysixICTtotal{}\,s per instance at $K=56$ on one core. While the comparison across software stacks is only indicative, the structural point is not: the SLqP subproblem, which previously required a general conic solver, is solved to a certified $10^{-9}$ gap with closed-form formulas and dense linear algebra, because the ICT surrogate is separable in the own signal and affine in the interference.

\begin{table}[t]
\centering\scriptsize\setlength{\tabcolsep}{4pt}
\caption{Time for one SLqP subproblem versus the number of links $K$ (medians over 3 instances, ms): dual Newton solve of Algorithm~\ref{alg:slqp} to a certified $10^{-9}$ gap (cold / warm start), and the QFT subproblem with Clarabel (interior point) and SCS (first order) at their default accuracies; $^\dagger$: at least one instance reported as inaccurate by the solver.}
\label{tab:scaling}
\begin{tabular}{rrrrrr}
\toprule
$K$ & $K_q$ & dual (cold) & dual (warm) & Clarabel & SCS\\
\midrule
\scalingTableRows
\bottomrule
\end{tabular}
\end{table}

\section{Conclusion}
\label{sec:concl}
In this paper, we have shown that every closed-form algorithm for weighted sum-rate maximization that we are aware of---closed-form FP, WMMSE, interference pricing, SEFP and XMMSE---maximizes a surrogate that is linear in the interference power, because that is the structure that decouples the links, and that this family has a unique tightest member, the interference-conjugate transform. Tightness forces the interference price, which explains why all of these methods charge the same one; the single-stream XMMSE surrogate coincides with that of SEFP; and the benchmarks lose everything they lose on the own-signal side, by curvature factors of $\gamma/2$ (FP/WMMSE) and $(\log\gamma)/2$ (SEFP/XMMSE) against $1$ for ICT. The resulting power-control, scheduling and MISO beamforming algorithms have the per-iteration cost of their FP counterparts, and the same structure yields the first solver-free algorithm for sum-least-$q$th-percentile rate maximization, with a certified duality gap in place of a conic solver. In our simulations, ICT needed the fewest iterations in every setting, reached equal or better stationary points, and under proportional-fairness scheduling improved the utility over SEFP by as much as SEFP improves over FP.

The ICT updates use the same information as interference pricing and closed-form FP, so that the algorithms are as amenable to distributed implementation as the benchmarks. Three limitations should be stated. (i) The guarantees require $p_{\min}>0$ (or a power floor for beamformers) to avoid the degenerate anchor $a_i=0$; this is shared with FP, WMMSE, SEFP and XMMSE. (ii) Dominance orders certified ascent and asymptotic rates rather than realized iterates; the simulations show that ICT reaches better points on average but not on every drop. (iii) Multi-stream users with $\log\det$ rates are not covered: the partial conjugate in the interference covariance does not collapse to a scalar function when the signal covariance has rank larger than one, and the quadratic XMMSE surrogate remains the best closed-form choice there; the tightest interference-affine surrogate for the $\log\det$ rate is the natural next step.

\appendices
\section{Proofs of Propositions~\ref{prop:curv} and \ref{prop:ascent}}
\label{app:curv}
\emph{Curvature.} Differentiating \eqref{eq:philq} and \eqref{eq:phirq} twice at $\nu=1$ gives $\phi_{\LQ}''(1)=-\gamma/2$ and $\phi_{\RQ}''(1)=-\log(1+\gamma)/2$ (only the $\sqrt\nu$ terms contribute). For \eqref{eq:phiict}, $\phi_{\ICT}''(1)=c^2\Om''(c)$ with $c=\gamma^2/(1+\gamma)$ and $\rhof(c)=\gamma$ by Proposition~\ref{prop:omega}(d), so that Proposition~\ref{prop:omega}(a) gives $\phi_{\ICT}''(1)=-\gamma/(\gamma+2)$. Dividing by $f''(1)=-\gamma^2/(1+\gamma)^2$ gives \eqref{eq:curv}, and $d^2/da^2=a_0^{-2}d^2/d\nu^2$ gives the expressions in the original variable.

\emph{Certified ascent.} $f(\mathbf x_1^+)\ge S_1(\mathbf x_1^+;\bar{\mathbf x})=\max_{\mathcal X}S_1\ge\max_{\mathcal X}S_2\ge S_2(\bar{\mathbf x};\bar{\mathbf x})=f(\bar{\mathbf x})$, by minorization, by $S_1\ge S_2$, and by tightness, respectively.

\section{Proof of Proposition~\ref{prop:rate}}
\label{app:rate}
Tightness at every anchor means $S_i(\mathbf x;\bar{\mathbf x})\le f(\mathbf x)$ with equality at $\mathbf x=\bar{\mathbf x}$, so that $f-S_i(\cdot;\bar{\mathbf x})$ has a minimum at $\bar{\mathbf x}$: $\nabla_{\mathbf x}S_i(\bar{\mathbf x};\bar{\mathbf x})=\nabla f(\bar{\mathbf x})$ and $\mathbf H_f\succeq\mathbf H_i$. Likewise, $S_1-S_2\ge0$ with equality at the anchor gives $\mathbf H_1\succeq\mathbf H_2$. Differentiating the identity $\nabla_{\mathbf x}S_i(\mathbf x;\mathbf x)=\nabla f(\mathbf x)$ in $\mathbf x$ gives $\partial^2S_i/\partial\mathbf x\partial\bar{\mathbf x}=\mathbf H_f-\mathbf H_i$ at $(\mathbf x^\star;\mathbf x^\star)$. The MM map is defined implicitly by $\nabla_{\mathbf x}S_i(T_i(\bar{\mathbf x});\bar{\mathbf x})=0$, so that $\mathbf H_i\,\mathbf J_i+(\mathbf H_f-\mathbf H_i)=0$, i.e., $\mathbf J_i=\mathbf I-\mathbf H_i^{-1}\mathbf H_f$. With $\mathbf P_i=(-\mathbf H_i)^{1/2}$, $\mathbf J_i$ is similar to $\mathbf I-\mathbf P_i^{-1}(-\mathbf H_f)\mathbf P_i^{-1}$, a symmetric matrix whose eigenvalues are $1-\lambda$ for the generalized eigenvalues $\lambda$ of $(-\mathbf H_f,-\mathbf H_i)$, all in $(0,1]$ because $0\prec-\mathbf H_f\preceq-\mathbf H_i$. The spectral radius is therefore $1-\lambda_{\min}$, with $\lambda_{\min}=\min_{\mathbf d}\mathbf d^T(-\mathbf H_f)\mathbf d/\mathbf d^T(-\mathbf H_i)\mathbf d$. Since $-\mathbf H_1\preceq-\mathbf H_2$, each Rayleigh quotient is larger for $i=1$; hence $\lambda_{\min}^{(1)}\ge\lambda_{\min}^{(2)}$ and $\varrho(\mathbf J_1)\le\varrho(\mathbf J_2)$. In SISO power control the interference terms of all three surrogates are affine in $\mathbf p$, so that $\mathbf H_i=\diag(w_j\phi_j''(a_j)g_{jj}^2)$.

\section{Proof of Proposition~\ref{prop:xmmse}}
\label{app:xmmse}
For one stream ($D_{i_k}=1$), the quantities of \cite[Sec.~III]{YangZhao26} become scalars: $\boldsymbol\Phi=\gamma$, the eigenvector is $1$, and $\tilde{\mathbf P}=\mathbf F^{-1}\mathbf H\mathbf V\mathbf P$ reduces to the scalar $c_0=z_0/b_0$ (in MISO, $\mathbf F=b_0$ and $\mathbf H\mathbf V=z_0=\mathbf h^H\bar{\mathbf v}$). The transmit-side coefficient matrices of \cite[eq.~(6)]{YangZhao26} evaluate to $\tilde{\mathbf A}=w|c_0|^2\tilde c(\gamma)\,\mathbf h\mathbf h^H$ with $\tilde c(\gamma)=\gamma^{-1}[\log(1+\gamma)/\gamma-1/(1+\gamma)]$, $\mathbf A_{i,l}=w|c_0|^2(1+\gamma)^{-1}\mathbf h_{il}\mathbf h_{il}^H$ and $\mathbf B=w\,c_0\,[\log(1+\gamma)/\gamma]\,\mathbf h$. The surrogate to be maximized in $\mathbf v$ is $2\Real(\mathbf B^H\mathbf v)-\mathbf v^H\tilde{\mathbf A}\mathbf v$ minus the interference terms, i.e.,
\begin{multline*}
\frac{2w\log(1+\gamma)}{\gamma}\Real(c_0^*z)-w|c_0|^2\tilde c(\gamma)|z|^2\\
-\sum_{i\ne j}\frac{w_i|c_{0,i}|^2}{1+\gamma_i}|\mathbf h_{ij}^H\mathbf v|^2 .
\end{multline*}
Using $|c_0|^2=\gamma/b_0$ and $|c_0|^2/(1+\gamma)=\gamma/(b_0(1+\gamma))=\kappa^\star$ gives the price $w\kappa^\star$ and, for the own-signal part, $\frac{2w\log(1+\gamma)}{\gamma b_0}\Real(z_0^*z)-\frac{w}{b_0}\big[\frac{\log(1+\gamma)}{\gamma}-\frac1{1+\gamma}\big]|z|^2$; since $\log(1+\gamma)/\gamma-1/(1+\gamma)=b(\gamma)/(\gamma(1+\gamma))$, this is \eqref{eq:xmmse-vec}. For RIT+QT, the complex QT applied to $w c(\gamma)|z|^2/(\gamma^2b+b(\gamma)|z|^2)$ gives $2\Real(y^*\sqrt{wc(\gamma)}\,z)-|y|^2(\gamma^2b+b(\gamma)|z|^2)$ with $y^\star=\sqrt{wc(\gamma)}z_0/(\gamma^2b_0+b(\gamma)|z_0|^2)=\sqrt w\,z_0/(\gamma b_0\sqrt{1+\gamma})$, where we have used $\gamma^2b_0+b(\gamma)|z_0|^2=\gamma b_0(1+\gamma)\log(1+\gamma)$. Then $2\Real(y^{\star*}\sqrt{wc}\,z)=\frac{2w\log(1+\gamma)}{\gamma b_0}\Real(z_0^*z)$, $|y^\star|^2b(\gamma)=\frac{wb(\gamma)}{\gamma(1+\gamma)b_0}$ and $|y^\star|^2\gamma^2=\frac{w\gamma}{(1+\gamma)b_0}=w\kappa^\star$, which is again \eqref{eq:xmmse-vec}. Setting $z=x$, $z_0=x_0$ real and $b=b_0$ gives the scalar identities, and subtracting the corresponding LDT+QT/WMMSE expression gives $(1-\log(1+\gamma)/\gamma)(x-x_0)^2$ in the notation of both papers.

\section{Proof of Proposition~\ref{prop:sched} and Validity of the Scheduling Surrogate}
\label{app:sched}
\emph{Validity.} For fixed $\lambda_i\ge0$ and any candidate $k\in\mathcal K_i$ with power $p_k$, Theorem~\ref{thm:ict} with $\kappa=\lambda_i/w_k$ gives $w_k\Om(\lambda_ig_{i,k}p_k/w_k)-\lambda_ib_i\le w_k\log(1+g_{i,k}p_k/b_i)$. Summing over cells gives a minorizer of the WSR for every schedule and power vector in which the price $\lambda_i$ multiplies $b_i=\sum_{j\ne i}g_{i,s_j}p_{s_j}+\sigma^2$; rearranging, user $k$ scheduled in cell $i$ pays $\mu_kp_k$, so that the minorizer is a sum over cells of $w_k\Om(\lambda_ig_{i,k}p_k/w_k)-\mu_kp_k$ and the cells decouple. With $\lambda_i=w_{s_i}\kappa_i^\star$ it is tight at the current schedule. Maximizing over $p_k$ gives $\lambda_ig_{i,k}/\rhof(\lambda_ig_{i,k}p_k/w_k)=\mu_k$, hence \eqref{eq:score}, with unclipped value $w_k[\log(1+\rho_k)+\rho_k/(1+\rho_k)]-\mu_kp_k=w_k\log(1+\rho_k)$ because $\mu_kp_k=w_k\rho_k/(1+\rho_k)$.

\emph{FP and SEFP.} In \cite{ShenYuII}, $y_i^2=w_{s_i}(1+\gamma_i)a_i/(a_i+b_i)^2=\lambda_i$; the candidate's reward $2y_i\sqrt{w_k(1+\gamma_i)g_{i,k}p_k}-y_i^2g_{i,k}p_k$ minus the penalty $\mu_kp_k$ is maximized at $p_k=w_k(1+\gamma_i)\lambda_ig_{i,k}/(\lambda_ig_{i,k}+\mu_k)^2$ with value $w_k(1+\gamma_i)\rho_k/(1+\rho_k)$, and adding the constant $w_k(\log(1+\gamma_i)-\gamma_i)$ gives $w_k\ell_{\mathrm{LDT}}(\rho_k;\gamma_i)$. In \cite{Jiao26}, $\alpha_i=\gamma_i$ and $\alpha_i^2y_i^2=\lambda_i$; the reward $2y_i\sqrt{w_kc(\gamma_i)g_{i,k}p_k}-b(\gamma_i)y_i^2g_{i,k}p_k$ minus $\mu_kp_k$ has maximum $w_kc(\gamma_i)\rho_k/(\gamma_i^2+b(\gamma_i)\rho_k)=w_k\ell_{\mathrm{RIT}}(\rho_k;\gamma_i)$. The ordering $\ell_{\mathrm{LDT}}\le\ell_{\mathrm{RIT}}\le\log(1+\cdot)$ is \cite[Prop.~1, Cor.~1]{Jiao26} and \eqref{eq:rit}, or Corollary~\ref{cor:dom} after maximizing over the power; the threshold in Section~\ref{sec:sched} is where $\ell_{\mathrm{LDT}}=0$.

\section{Proof of the Statements in Remark~\ref{rem:mimo}}
\label{app:mimo}
$R(\boldsymbol\alpha,\mathbf B)=\log\det(\mathbf B+\boldsymbol\alpha\boldsymbol\alpha^H)-\log\det\mathbf B$ is convex in $\mathbf B\succ0$: along $\mathbf B+s\boldsymbol\Delta$, $\partial_s^2R=\tr[(\mathbf B^{-1}\boldsymbol\Delta)^2]-\tr[((\mathbf B+\boldsymbol\alpha\boldsymbol\alpha^H)^{-1}\boldsymbol\Delta)^2]\ge0$, because $\mathbf X\mapsto\tr(\mathbf X\boldsymbol\Delta\mathbf X\boldsymbol\Delta)$ is monotone on the positive semidefinite cone and $\mathbf B^{-1}\succeq(\mathbf B+\boldsymbol\alpha\boldsymbol\alpha^H)^{-1}$. A minorizer $\phi(\boldsymbol\alpha)-\tr(\mathbf K\mathbf B)$ tight at $(\boldsymbol\alpha_0,\mathbf B_0)$ makes $\mathbf B\mapsto R(\boldsymbol\alpha_0,\mathbf B)+\tr(\mathbf K\mathbf B)$ attain its minimum at the interior point $\mathbf B_0$, so that $\mathbf K=-\nabla_{\mathbf B}R(\boldsymbol\alpha_0,\mathbf B_0)=\mathbf K_0$. With $\mathbf u_0=\mathbf B_0^{-1}\boldsymbol\alpha_0$ and $t=\mathbf u_0^H\mathbf B\mathbf u_0$, $\tr(\mathbf K_0\mathbf B)=t/(1+\gamma_0)$, and Cauchy--Schwarz gives $|\mathbf u_0^H\boldsymbol\alpha|^2\le t\,\boldsymbol\alpha^H\mathbf B^{-1}\boldsymbol\alpha$; hence $R+\tr(\mathbf K_0\mathbf B)\ge\log(1+|\mathbf u_0^H\boldsymbol\alpha|^2/t)+t/(1+\gamma_0)\ge\Om(|\mathbf u_0^H\boldsymbol\alpha|^2/(1+\gamma_0))$ by Theorem~\ref{thm:ict}. Equality is approached by $\mathbf B_\epsilon=c\,\boldsymbol\alpha\boldsymbol\alpha^H+\epsilon\mathbf I$, $\epsilon\to0$, with $c$ chosen so that $t$ equals the minimizer of the scalar problem, so that the infimum equals $\Om(\cdot)$.

\section*{Acknowledgment}
The interference-conjugate transform, its tightness theorem, the algorithms and the code of this paper were developed in a custom workflow in which Claude Fable 5.1 (Anthropic) was directed by the author toward the final transform. In the author's view, this indicates that AI systems can now take part substantively in mathematical discovery in signal processing and optimization---identifying the right abstraction, proving optimality within it, and turning it into working algorithms---when guided by a researcher who asks the right questions.

\end{document}